%% file: pseudonet.tex
\documentclass[11pt]{article}
\pdfoutput=1

\usepackage[utf8]{inputenc} 
\usepackage[T1]{fontenc}    
\usepackage{hyperref}       
\usepackage{url}            
\usepackage{booktabs}       
\usepackage{amsfonts}       
\usepackage{nicefrac}       
\usepackage{microtype}      
\usepackage{setspace}
\usepackage[round,authoryear]{natbib}
\usepackage{titling}
\usepackage{enumerate}

\input{alnur}
\usepackage{multirow}
\usepackage[export]{adjustbox}
\usepackage{fullpage}
\usepackage{multirow}
\usepackage[export]{adjustbox}
\usepackage{algorithm}
\usepackage{algorithmic}
\usepackage{subcaption}
\usepackage{xr}
\newcommand{\conc}{CONCORD}                     
\newcommand{\ours}{PseudoNet}
\newcommand{\glasso}{GLasso}
\newcommand{\spacee}{SPACE}
\newcommand{\splice}{SPLICE}

\newcommand{\condreg}{CondReg}
\newcommand{\samplee}{Sample}
\newcommand{\ledoit}{Ledoit}
\newcommand{\djia}{DJIA}

\newcommand{\concs}{\textrm{con}}                 
\newcommand{\ourss}{\textrm{net}}
\newcommand{\spls}{\textrm{spl}}
\newcommand{\spcs}{\textrm{spc}}
\newcommand{\X}{\textrm{off}}
\newcommand{\D}{\textrm{diag}}

\newcommand{\Omegatrue}{\Omega^0}
\newcommand{\bic}{\mathbf{Bic}}
\newcommand{\rss}{\mathbf{rss}}
\def\P{\mathbb{P}}
\renewcommand{\P}{\prob}

\title{Generalized Pseudolikelihood Methods for Inverse Covariance Estimation}
\author{Alnur Ali$^1$ \and Kshitij Khare$^2$ \and Sang-Yun Oh$^3$ \and Bala Rajaratnam$^4$}
\date{$^1$Machine Learning Department, Carnegie Mellon University \\ 
$^2$Department of Statistics, University of Florida \\ 
$^3$Department of Statistics and Applied Probability, UC Santa Barbara \\
$^4$Department of Statistics, Stanford University}

\begin{document}
\maketitle
\begin{abstract}
\input{abs}
\end{abstract}

\newif\ifwantbiomtabs
\wantbiomtabsfalse
\newcommand{\fsw}{0.31\textwidth}

\input{intro}
\input{ours}
\input{synth}
\input{finance}
\input{neuro}
\input{theory}
\input{disc}

\bibliographystyle{plainnat}
\bibliography{refs}

\title{Supplement to ``Generalized Pseudolikelihood Methods for Inverse Covariance Estimation''}
\maketitle

\renewcommand\thesection{S.\arabic{section}}  
\renewcommand\theequation{S.\arabic{equation}}
\renewcommand\thetable{S.\arabic{table}}
\renewcommand\thefigure{S.\arabic{figure}}

\input{supp_screen}
\input{supp_finance}
\input{supp_conv}
\input{supp_sat}
\input{supp_cons}
\input{supp_diags}

\end{document}

%% file: alnur.tex
\newcommand{\ones}{\mathbf 1}
\newcommand{\reals}{{\mbox{\bf R}}}
\newcommand{\symm}{{\mbox{\bf S}}}  

\newcommand{\Tr}{\mathop{\bf Tr}}

\newcommand{\card}{\mathop{\bf card}}

\newcommand{\prox}{\mathbf{prox}}

\newcommand{\Expect}{\mathop{\bf E{}}}

\newcommand{\prob}{\mathop{\bf Pr}} 

\newcommand{\argmin}{\mathop{\rm argmin}}

\newcommand{\sign}{\mathop{\bf sign}}

\newcommand{\cf}{{\it cf.}}

\newcommand{\ie}{{\it i.e.}}
\newcommand{\etc}{{\it etc.}}

\usepackage{graphicx,psfrag,amsmath,amsfonts,verbatim} 
\usepackage{amsmath,amsfonts,verbatim}
\usepackage[small,bf]{caption}
\usepackage{url}

\ifdefined\aaslides
\else

\usepackage{amsthm}
\newtheorem{theorem}{Theorem}[section]
\newtheorem{lemma}[theorem]{Lemma}

\newtheorem{corollary}[theorem]{Corollary}

\fi

\newcommand{\minimize}{\mathop{\mbox{minimize}}} 
\newcommand{\minimizewrt}[1]{\underset{#1}{\minimize}}

\newcommand{\subjectto}{\mbox{subject to}}

\usepackage{xcolor}

\newcommand{\optprobstart}{\begin{equation}}
\newcommand{\optprobend}{\end{equation}}

\newcommand{\optprobstartnn}{\begin{equation*}}
\newcommand{\optprobendnn}{\end{equation*}}

\usepackage{listings}

\usepackage{color}

\newcommand{\mtxstart}{\left[\begin{array}}
\newcommand{\mtxend}{\end{array}\right]}	

\usepackage{tikz}

\newcommand{\vect}{\mathop{\bf vec}}
\newcommand{\vech}{\mathop{\bf vech}}


%% file: abs.tex
We introduce \ours, a new \textit{pseudolikelihood}-based estimator of the inverse covariance matrix, that has a number of useful statistical and computational properties.  We show, through detailed experiments with synthetic and also real-world finance as well as wind power data, that \ours~outperforms related methods in terms of estimation error and support recovery, making it well-suited for use in a downstream application, where obtaining low estimation error can be important.  We also show, under regularity conditions, that \ours~is consistent.  Our proof assumes the existence of accurate estimates of the diagonal entries of the underlying inverse covariance matrix; we additionally provide a two-step method to obtain these estimates, even in a high-dimensional setting, going beyond the proofs for related methods.  Unlike other pseudolikelihood-based methods, we also show that \ours~does not \textit{saturate}, \ie, in high dimensions, there is \textit{no} hard limit on the number of nonzero entries in the \ours~estimate.  We present a fast algorithm as well as \textit{screening rules} that make computing the \ours~estimate over a range of tuning parameters tractable.

%% file: intro.tex
\section{Introduction}

In this paper, we consider the problem of obtaining a sparse estimate of the inverse covariance matrix of a collection of random variables in a high-dimensional setup, where the number of variables (\ie, features) $p$ is possibly much larger than the number of data samples $n$.  This is an important problem in modern statistics as well as across a variety of applications, including finance (see, for example, \citet{ledoit2003honey,yuan2007,won2013condition,khare2014convex}) and biology (see, for example, \citet{BEA:08,friedman2008sparse,rothman2008,peng2009partial,friedman2010applications,khare2014convex}).  In many cases, the obtained estimate is used in a downstream application in some way, and the sparsity pattern of the estimate is often inspected and interpreted, in order to reveal the nature of the conditional independencies between the random variables.  Sparsity is useful here for a number of reasons, including making the resulting estimates more interpretable, especially in high dimensions, where we would like the number of nonzero entries in our estimate to be small.  

In high dimensions (\ie, when $p \gg n$), it makes sense to obtain an estimate by maximizing an $\ell_1$-penalized Gaussian likelihood (see, for example, \citet{yuan2007,BEA:08,friedman2008sparse,rothman2008}) --- although other penalities are certainly possible.  This is, of course, a massive area of research, and a number of estimators for, as well as extensions to, this basic Gaussian setup have been proposed over the years, including the seminal graphical lasso algorithm (\glasso) of \citet{friedman2008sparse}.  \textit{Pseudolikelihood}-based estimators \citep{besag1974} take a somewhat different approach, in that they can be seen as (roughly) minimizing the sum of a collection of $\ell_1$-penalized regression (\ie, lasso) problems, one for each variable, which more directly exploits the connection between the inverse covariance matrix and partial correlations; see, for example, \citet{meinshausen2006high,rocha2008path,peng2009partial,friedman2010applications,khare2014convex,alnur}.  Pseudolikelihood-based estimators are thus, in a sense, simpler and more flexible in moving beyond the usual Gaussian setup than other estimators.

Under the assumption that the data-generating process is multivariate normal, it is a well-known fact that the random variables $i$ and $j$ are conditionally independent given the remaining variables if and only if the $(i,j)$ entry in the underlying inverse covariance matrix is zero (see, for example, \citet{lauritzen1996graphical}); this fact is often used to obtain an undirected graphical model of the data, where the vertices of an undirected graph are put in one-to-one correspondence with the random variables, and the absence of an edge between any two vertices takes on the special meaning that the corresponding variables are conditionally independent given the remaining variables.  As a result, much work has looked at producing estimates that accurately recover the underlying support (\ie, the set of nonzero entries) --- on the other hand, we often want to use an estimate later in our workflow, in which case low estimation error (as measured by a suitable matrix norm) is perhaps a more useful criterion for evaluating an estimate.  Asymptotically, the \spacee~and \conc~pseudolikelihood-based estimators of \citet{peng2009partial} and \citet{khare2014convex}, respectively, have been shown to be consistent (in a Frobenius norm sense) under certain conditions; however, carefully checking the conditions required by the consistency proofs in these papers reveals that they presume the existence of accurate estimates of the diagonal entries of the underlying inverse covariance matrix.  A natural choice here is to simply use the diagonal entries of the sample inverse covariance matrix, but such estimates unfortunately do not exist when $p > n$, and alternatives are not immediately apparent.

Returning to the issue of interpretability of pseudolikelihood-based estimates, we raise a basic question: are the estimates given by pseudolikelihood-based methods well-defined (\ie, unique)?  We elaborate below (see Section \ref{sec:related}), but the short answer to this question for now is that the estimates given by many pseudolikelihood-based methods, including \spacee, \conc, the \splice~estimator of \citet{rocha2008path}, as well as the Symmetric Lasso estimator of \citet{friedman2010applications}, may not be unique, and in fact many of these methods may not even converge to a particular estimate --- which can be problematic from an interpretability point of view.  For example, in a finance application, we may wish to understand which assets are correlated, in order to assemble a well-diversified portfolio \citep{markowitz1952portfolio}; if the outcome of an estimation procedure is not necessarily unique, then which estimate/assets should we use?

Furthermore, given the connection between pseudolikelihood-based methods and the lasso, we recall a basic result from lasso theory, which states that the lasso can \textit{saturate}, meaning that when $p > n$, there exists a lasso estimate with at most $n$ nonzero entries (equivalently, selected variables) \citep{rosset2004boosting,Zou05regularizationand,tibshirani2013}; this behavior can be quite limiting from the points of view of interpretability as well as estimation error.  It is therefore natural to ask: do estimates given by existing pseudolikelihood-based methods also saturate?  We show that several estimators, including \spacee, \conc, and SPLICE, unfortunately can saturate, which establishes an analogous result for undirected graphical models (see Section \ref{sec:sat}).

\subsection{Overview of contributions}

In this paper, we introduce a new, more flexible pseudolikelihood-based estimator of the inverse covariance matrix, which we call \ours, that addresses all the aforementioned issues with existing pseudolikelihood-based methods, while preserving their useful properties.  Additionally, the \ours~estimator possesses a number of other useful statistical and computational properties.  We give a brief summary below.

\begin{itemize}
\item \textit{Computational aspects and uniqueness}.  We present a fast algorithm for computing the \ours~estimate, by leveraging recent advances in convex optimization; our algorithm runs in just a few seconds on a standard laptop\footnote{In more detail, the laptop we use is a standard 2015 MacBook Pro, with a two core 3.1 GHz Intel Core i7 5557C processor and 16 GB of memory.}, for problems with thousands of variables.  We show that our algorithm converges at a geometric (``linear'') rate to the (global) solution of a convex optimization problem that defines the \ours~estimate.  Furthermore, this solution is unique, as the objective in the optimization problem is strictly convex.  This contrasts with a number of other pseudolikelihood-based methods \citep{rocha2008path,peng2009partial,friedman2010applications,khare2014convex,NIPS2014_5576}, which do not provide unique estimates, making interpretation difficult, and additionally are either not guaranteed to converge or converge at a slower rate, as with the \conc~estimator of \citet{NIPS2014_5576}.

We also derive \textit{screening rules} for \ours~\citep{BEA:08,tibshirani2012strong,mazumder2012exact}, by leveraging the precise nature of the \ours~optimization problem, which make the optimization problem much faster to solve by omitting some of the variables.  These rules can be implemented as simple checks based on the optimality conditions of the \ours~optimization problem; in some cases, we are able to reduce the size of the optimization problem by 90\%.

\item \textit{Estimation error}.  We show, through detailed experiments with synthetic data, that \ours~significantly outperforms the closely related \conc~estimator of \citet{khare2014convex} --- that we build upon --- in terms of estimation error (as measured by several matrix norms), while also outperforming \conc~in terms of support recovery (\ie, variable selection).  As mentioned above, although the literature often emphasizes support recovery, obtaining an estimate with low estimation error is perhaps more useful in situations where our estimate will be used by a downstream application.

\item \textit{Consistency}.  We also show, under standard regularity conditions, that \ours~is consistent at a rate of $\sqrt{(\log p) / n}$.  The consistency proofs for the related pseudolikelihood-based estimators \spacee~and \conc~assume the existence of accurate estimates of the diagonal entries of the underlying inverse covariance matrix, but do not provide a method for obtaining these estimates when $p > n$.  In this paper, we go further and give a two-step method that obtains accurate diagonal estimates, even when $p > n$; this result is therefore also useful in the consistency proofs for \spacee~\citep[Theorem 3]{peng2009partial} and \conc~\citep[Theorem 2]{khare2014convex}.

\item \textit{Saturation}.  We show that the \ours~estimate does \textit{not} saturate, meaning that when $p \gg n$, the number of variables selected by \ours~can be greater than $np$ (out of $p(p-1)/2$ total variables), which is not true for several other pseudolikelihood-based estimators \citep{rocha2008path,peng2009partial,khare2014convex}; establishing this result involves generalizing an analogous claim for the (standard) lasso as in, for example, \citet{rosset2004boosting,tibshirani2013}.  This result is useful from the points of view of the estimation error as well as the interpretability of the \ours~estimate.

\item \textit{Non-Gaussian data}.  Lastly, we illustrate, through numerical examples with real-world finance and wind power data, that \ours~deals effectively with non-Gaussian data, outperforming several strong baselines.  This is due, in part, to the fact that the precise form of the objecive in the \ours~optimization problem dispenses with the assumption that the true distribution is normal, which is helpful in moving beyond the usual Gaussian setup.
\end{itemize}

\subsection{Outline}

An outline for the rest of this paper is as follows.  In the next subsection, we survey related work.  In Section \ref{sec:ours}, we describe the \ours~estimator and its screening rules.  In Section \ref{sec:exps}, we present an empirical evaluation of \ours, as well as several baselines, on synthetic and real-world data.  We present all of our theoretical results on \ours's statistical and computational properties in Section \ref{sec:theory}; all of our proofs are given in the supplement.  We conclude with a brief discussion in Section \ref{sec:disc}.

\subsection{Related work}
\label{sec:related}

The literature on high-dimensional sparse inverse covariance estimation is quite vast; we do not claim to give a complete treatment of it here, and instead highlight work most related to our own.  \citet{yuan2007,BEA:08,friedman2008sparse,rothman2008} first proposed estimating the inverse covariance matrix by maximizing an $\ell_1$-penalized Gaussian likelihood; \citet{friedman2008sparse}, in particular, proposed the \glasso, a fast algorithm for computing an estimate in this framework.  In a related but distinct line of work, a number of pseudolikelihood-based estimators have been proposed; pseudolikelihood-based methods take a somewhat different perspective, in that they can be seen as roughly minimizing a series of $\ell_1$-penalized regression problems, making them arguably simpler to analyze and extend than other approaches.  The seminal \textit{neighborhood selection} method of \citet{meinshausen2006high}, which fits a lasso regression of each variable on the rest, is an example; a drawback of neighborhood selection, however, is that the neighborhood selection estimate may not be symmetric, so a post-processing step is required.

In a nice step forward, \citet{peng2009partial} introduced the \spacee~estimator, and showed that it is symmetric and also consistent, under suitable regularity conditions.  Unfortunately, \spacee~is not guaranteed to converge (it is easy to find examples where the iterates produced by \spacee~alternate between two values), and furthermore the \spacee~estimate may not be unique \citep{khare2014convex}; additionally, the consistency proof for \spacee~assumes that accurate estimates for the diagonal entries of the underlying inverse covariance matrix are available, even when $p > n$, without giving a method to obtain them.  Inspired by \spacee, \citet{friedman2010applications} introduced the Symmetric Lasso estimator, which is also symmetric, but is not guaranteed to converge, be unique, or be consistent \citep[Lemma 2]{khare2014convex}.  The \splice~estimator of \citet{rocha2008path} has some useful computational properties, but unfortunately does not have any of these guarantees either \citep[Lemma 3]{khare2014convex}.

Building on \spacee, the \conc~estimator \citep{khare2014convex,NIPS2014_5576} recently made useful progress: \conc~is symmetric, like \spacee, but is additionally guaranteed to converge at a rate of $O(1/k^2)$, where $k$ here is the number of iterations, and is also consistent.  On the downside, as we show later in this paper, \conc's consistency proof assumes accurate diagonal estimates even when $p > n$, its estimate may not be unique when $p > n$, and it can saturate (\ie, when $p \gg n$, the \conc~estimate can select at most $np$ out of $p(p-1)/2$ total variables).

%% file: ours.tex
\section{The \ours~estimator}
\label{sec:ours}

Assume that we are given $n$ samples $X_{1 \cdot},\ldots,X_{n \cdot} \in \reals^p$, drawn i.i.d.~from some unknown distribution that, without a loss of generality, we take to have mean zero and covariance matrix $\Sigma^0 \in \symm_{++}^p$ (the space of $p \times p$ positive definite matrices).  We want to estimate the underlying inverse covariance matrix $\Omegatrue = (\Sigma^0)^{-1}$ with a small number of nonzero entries.

We define the \ours~estimate, which gives a sparse estimate of the underlying inverse covariance matrix, as the solution of the following convex optimization problem:
\begin{equation*}
\begin{array}{ll}
\minimizewrt{\Omega \in {\bf R}^{p \times p}} & - (1/2) \sum_{i=1}^p \log (\Omega_{ii}^2) + (1/2) \sum_{i=1}^p \left\| \Omega_{ii} X_i + \sum_{j \neq i}^p \Omega_{ij} X_j \right\|_2^2 \\
& \quad\quad + \lambda_1 \sum_{i \neq j}^p | \Omega_{ij} | + (\lambda_2 / 2) \| \Omega \|_F^2,
\end{array}
\end{equation*}
where $\lambda_1, \lambda_2 > 0$ are tuning parameters, and $\| \cdot \|_F$ is the Frobenius norm.  After some manipulations, we can put the above optimization problem into the following matrix form, which is useful for much of the remainder of the paper:
\optprobstart
\begin{array}{ll}
\minimizewrt{\Omega \in {\bf R}^{p \times p}} & - (1/2) \log \det(\Omega_\D^2) + (n/2) \Tr S \Omega^2 + \lambda_1 \| \Omega_\X \|_1 + (\lambda_2 / 2) \| \Omega \|_F^2.
\end{array}
\label{eq:ours}
\optprobend
Here, $\Omega_\D \in \reals^{p \times p}$ is a matrix of the diagonal entries of $\Omega$, with its off-diagonal entries set to zero; $S \in \reals^{p \times p}$ is the sample covariance matrix, \ie, $S = (1/n) X^T X$, and $X \in \reals^{n \times p}$ is a data matrix; $\Omega_\X \in \reals^{p \times p}$ is a matrix of the off-diagonal entries of $\Omega$, with its diagonal entries set to zero; and $\| \cdot \|_1$ is the elementwise $\ell_1$ norm.

Note that we do not make the assumption here that the underlying data-generating process is, for example, multivariate normal, which is helpful in moving beyond the usual Gaussian setup; nonetheless, the objective of the \ours~optimization problem in matrix form \eqref{eq:ours} does bear some resemblance to an $\ell_1$-penalized Gaussian likelihood.  In fact, the \ours~optimization problem \eqref{eq:ours} generalizes the (standard) $\ell_1$-penalized Gaussian maximum likelihood problem (by design), when \eqref{eq:ours} is written as
\begin{equation*}
\begin{array}{ll}
\minimizewrt{\Omega \in {\bf R}^{p \times p}} & -(1/2) \log \det F(\Omega) + (n/2) \Tr S G(\Omega) + \lambda_1 \| H(\Omega) \|_1 + (\lambda_2 / 2) \| \Omega \|_F^2,
\end{array}
\label{eq:gen}
\end{equation*}
for some operators $F,G,H:\reals^{p \times p} \rightarrow \reals^{p \times p}$.  (Taking $F$ as $\Omega \mapsto \Omega_\D^2$, $G$ as $\Omega \mapsto \Omega^2$, and $H$ as $\Omega \mapsto \Omega_\X$ recovers the \ours~optimization problem \eqref{eq:ours}.)  Now taking $F$, $G$, and $H$ all as $\Omega \mapsto \Omega$, with $\lambda_2 = 0$, recovers the \glasso~optimization problem \citep[Equation 1]{friedman2008sparse}.  Furthermore, the framework above also generalizes several pseudolikelihood-based approaches; for example, taking $F$ as $\Omega \mapsto \Omega_\D$, $G$ as $\Omega \mapsto \Omega \Omega_\D^{-1} \Omega$, $H$ as $\Omega \mapsto \Omega_\X$, and $\lambda_2 = 0$ recovers the \spacee~optimization problem \citep[Equation 2]{peng2009partial}, and taking $F$ as $\Omega \mapsto \Omega_\D^2$, $G$ as $\Omega \mapsto \Omega^2$, $H$ as $\Omega \mapsto \Omega_\X$, and $\lambda_2 = 0$ recovers the \conc~optimization problem \citep[Equation 8]{khare2014convex}, revealing a close connection between the \ours~and \conc~optimization problems.

Although simple in appearance, the squared Frobenius norm penalty in the \ours~optimization problem \eqref{eq:ours} gives \ours~a number of statistical and computational advantages (that are not always simple to show) over many other pseudolikelihood-based approaches, including the ones just mentioned.\footnote{Some care is also required here: the theory that we develop in this paper does not necessarily follow if, for example, $\| \Omega_\X \|_F^2$ is used instead of $\| \Omega \|_F^2$.}  Statistically, owing to this penalty, \ours~is able to obtain much better estimation error than \conc~(see Sections \ref{sec:exps}, \ref{sec:cons}, and \ref{sec:diags}), which is again useful when our estimate will be used by a downstream application; the estimates produced by \ours~also tend to be more stable than those produced by \conc.  We can understand this intuitively, by considering the relationship between the \textit{elastic net} \citep{Zou05regularizationand} and the (standard) lasso optimization problems: the elastic net augments the objective in the lasso optimization problem with a  ridge penalty, which is seen as giving a sparse estimate with better prediction error than the associated lasso estimate --- taking a pseudolikelihood-based approach makes it natural to incorporate these ridge penalties into each regression (sub)problem, in order to obtain a sparse estimate of the inverse covariance matrix with low estimation error.

The elastic net is also an elegant solution to the issue of saturation in the lasso (\ie, when $p > n$, the number of variables selected by the lasso can be at most $n$).  Even though pseudolikelihood-based estimators and the lasso are connected in many ways, it is still natural to wonder if pseudolikelihood-based estimators can also saturate, since the objectives in the defining optimization problems for many pseudolikelihood-based estimators include terms that go beyond pure lasso regressions?  We show later (see Section \ref{sec:sat}) that several pseudolikelihood-based estimators (specifically, \splice, \spacee, and \conc) indeed can saturate --- and that the squared Frobenius norm penalty in the \ours~optimization problem \eqref{eq:ours} is what prevents it from saturating.  This is a useful result for \ours, from the points of view of the estimation error as well as the interpretability of the \ours~estimate.

Finally, the choices of $F$, $G$, and $H$ that we make in the general framework above in order to arrive at the the \ours~optimization problem \eqref{eq:ours} ensure that \eqref{eq:ours} is convex; further imposing the squared Frobenius norm penalty guarantees that the objective in \eqref{eq:ours} is strictly convex, and hence the \ours~estimate is always unique (as mentioned above, convexity as well as uniqueness are not guaranteed for many other pseudolikelihood-based estimators).  Computationally, the squared Frobenius norm penalty also allows us to derive a fast algorithm for computing the \ours~estimate (which we do next) that converges to the unique, global solution of the \ours~optimization problem \eqref{eq:ours} at a geometric rate (see Section \ref{sec:conv}), and is much faster than \conc~(see Section \ref{sec:exps}).

Next, we turn to deriving a fast algorithm for computing the \ours~estimate.  Rewriting \eqref{eq:ours} as the sum of a smooth function $g$ and a nonsmooth function $h$, \ie, letting $f(\Omega)$ be the objective in \eqref{eq:ours}, we have that $f(\Omega) = g(\Omega) + h(\Omega)$, with
\begin{equation}
g(\Omega) = - (1/2) \log \det(\Omega_\D^2) + (n/2) \Tr S \Omega^2 + (\lambda_2 / 2) \| \Omega \|_F^2, \quad h(\Omega) = \lambda_1 \| \Omega_\X \|_1. \label{eq:prox2}
\end{equation}
The presence of the nonsmooth term $h$ here makes the \ours~optimization problem \eqref{eq:ours} difficult to solve using, say, an interior point method.  On the other hand, $h$ does admit a computationally efficient \textit{proximal operator} \citep{parikh2013proximal}, \ie,
\begin{align}
\prox_{t h}(V) & =
\argmin_{Z \in {\bf R}^{p \times p}} \left( h(Z) + \frac{1}{2t} \| Z-V \|_F^2 \right) \notag \\
\implies \left[ \prox_{t h}(V) \right]_{ij} & =
\begin{cases}
V_{ij} - t & V_{ij} > t \\
0 & | V_{ij} | \leq t \\
V_{ij} + t & V_{ij} < -t
\end{cases}
\quad i,j=1,\ldots,p, \label{eq:prox}
\end{align}
for some $V \in \reals^{p \times p}$ and constant $t > 0$; \eqref{eq:prox} is known as the (elementwise) \textit{soft-thresholding} operator.  Thus, a proximal gradient method\footnote{It is straightforward derive an accelerated proximal gradient method as well.} is a natural choice here; \ie, on each iteration of the algorithm, we take a step in the direction of the negative gradient of $g$, and then apply \eqref{eq:prox}.  Provided that the gradient of $g$ is Lipschitz continuous and the step sizes are chosen appropriately, proximal gradient methods in general obtain a convergence rate of $O(1/k)$, where $k$ here is the number of iterations.  However, we are able to obtain a much better (\ie, geometric) rate of convergence, owing to the strong convexity of \eqref{eq:ours}, as we show later in Section \ref{sec:conv}.

To complete the specification of the proximal gradient method, we give the gradient and Hessian of the smooth term $g$ in \eqref{eq:prox2}:
\begin{align}
\nabla g(\Omega) & = -\Omega_\D^{-1} + (n/2) (S \Omega + \Omega S) + \lambda_2 \Omega \label{eq:grad} \\
\nabla ^2 g(\Omega) & = \sum_{i=1}^p (1/\Omega_{ii}^2) ( e_i e_i^T \otimes e_i e_i^T ) + (n/2) ( S \otimes I + I \otimes S ) + \lambda_2 I_{p^2}, \label{eq:hess}
\end{align}
where $\otimes$ denotes the Kronecker product, and $e_i$ denotes the $i$th standard basis vector in $\reals^p$.  The complete algorithm for computing the \ours~estimate is specified in Algorithm \ref{alg:ours}; assuming the iterates are sparse, the computational cost of each iteration of Algorithm \ref{alg:ours} is dominated by computing the soft-thresholding operator, and therefore costs $O(p^2)$.

\newlength\myindent
\setlength\myindent{2em}
\newcommand\bindent{%
  \begingroup
  \setlength{\itemindent}{\myindent}
  \addtolength{\algorithmicindent}{\myindent}
}
\newcommand\eindent{\endgroup}

\newcommand{\pushcode}[1][1]{\hskip\dimexpr#1\algorithmicindent\relax}

\begin{algorithm}
\caption{Proximal gradient method for computing the \ours~estimate}
\label{alg:ours}
\begin{algorithmic}
	\STATE \textbf{Input:} data matrix $X \in \reals^{n \times p}$, tuning parameters $\lambda_1,\lambda_2 > 0$
	\STATE \textbf{Output:} estimate $\hat{\Omega}^\ourss$

	\STATE \textbf{initialize} starting point $\Omega \in \symm_{++}^p$ (the space of $p \times p$ positive definite matrices); optimization tolerance $\epsilon > 0$; line search parameters $\tau_{\textrm{init}}, \beta \in (0,1)$

	\REPEAT
		\STATE compute $\nabla g(\Omega)$ using Equation \ref{eq:grad}
		\STATE choose $\tau$ via backtracking line search as follows
			\bindent
			\STATE set $\tau \leftarrow \tau_{\textrm{init}}$
			\STATE set $\tilde{\Omega} \leftarrow \prox_{(\lambda_1 \tau) h}(\Omega - \tau \nabla g(\Omega))$ using Equation \ref{eq:prox}
			\WHILE{$g(\tilde{\Omega}) \geq g(\Omega) - \Tr \left( (\nabla g(\Omega))^T (\Omega - \tilde{\Omega}) \right) + \frac{1}{2\tau} \| \Omega - \tilde{\Omega} \|_F^2$}
			\STATE \pushcode[-0.3] ($\| \cdot \|_F$ is the Frobenius norm)
			\STATE update $\tilde{\Omega} \leftarrow \prox_{(\lambda_1 \tau) h}(\Omega - \tau \nabla g(\Omega))$
			\STATE update $\tau \leftarrow \beta \tau_{\textrm{init}}$
			\ENDWHILE
			\STATE output $\tau$
			\eindent
		\STATE update $\Omega \leftarrow \prox_{(\lambda_1 \tau) h}(\Omega - \tau \nabla g(\Omega))$
	\UNTIL{stopping criterion is satisfied, \ie, until $\| \nabla g(\Omega) + z \|_F / \| \Omega \|_F \leq \epsilon$ ($z$ is any subgradient of $h$ evaluated at $\Omega$)}
	\STATE output $\hat{\Omega}^\ourss = \Omega$
\end{algorithmic}
\end{algorithm}

\subsection{Choice of tuning parameters}

Next, we provide a way to choose the tuning parameters $\lambda_1$ and $\lambda_2$ in the \ours~optimization problem \eqref{eq:ours}.  We propose choosing these parameters by selecting the $(\lambda_1, \lambda_2)$ pair that minimizes the following Bayesian information criterion-like score over a grid of tuning parameter values:
\begin{equation}
\bic(\lambda_1, \lambda_2) = \sum_{j=1}^p \bic(\lambda_1, \lambda_2, j),
\label{eq:bic}
\end{equation}
where
\begin{align*}
\bic(\lambda_1, \lambda_2, j) & = n \log \rss(\lambda_1, \lambda_2, j) + \log n \times \left| \left\{ \ell : \ell \in \{1,\ldots,p\}, \; \ell \neq j, \; \hat{\Omega}_{j \ell}^\ourss(\lambda_1, \lambda_2) \neq 0 \right\} \right|, \\
\rss(\lambda_1, \lambda_2, j) & = \sum_{i=1}^n \left( X_{ij} - \sum_{k \neq j}^p \left( \hat{\Omega}^\ourss_{j k}(\lambda_1, \lambda_2) / \hat{\Omega}^\ourss_{jj}(\lambda_1, \lambda_2) \right) X_{i k} \right)^2,
\end{align*}
and $\hat{\Omega}^\ourss(\lambda_1, \lambda_2)$ is the solution of the \ours~optimization problem \eqref{eq:ours} for a particular $\lambda_1$ and $\lambda_2$.  This method is simple to implement and computationally inexpensive, especially when combined with the screening rules that we describe in the next subsection.

\subsection{Omitting predictors via screening rules}

We often want to solve the \ours~optimization problem \eqref{eq:ours} over a grid of $(\lambda_1, \lambda_2)$ values, and then choose a suitable estimate (for example, by using the procedure outlined in the previous subsection).  By leveraging the particular form of the \ours~optimization problem, we derive \textit{sequential strong} screening rules here \citep{tibshirani2012strong}, which are well-suited for this, because they omit variables from the \ours~optimization problem as we solve it over a range of tuning parameter values.

\citet{tibshirani2012strong} introduced sequential strong screening rules as a framework for deriving screening rules that drop variables as we solve a sequence of convex optimization problems; these optimization problems are required to have an objective that can be expressed as the sum of a smooth loss and a potentially nonsmooth penalty.  Sequential strong rules are based on the optimality conditions for the optimization problem in question, as well as the assumption that the gradient of the smooth loss is \textit{nonexpansive}, \ie, that it has a Lipschitz constant equal to one; thus, strong rules might commit \textit{violations}, \ie, they might suggest that a variable could be dropped when it is actually nonzero at the solution.  Consequently, we (usually) check the optimality conditions after applying sequential strong rules; we do so in our numerical experiments, and never observe a violation (see Sections \ref{sec:synth} and \ref{sec:neuro}).

Sequential strong rules build on the work of \citet[Theorem 4]{BEA:08}, who first observed that variables can be dropped from their particular optimization problem by arguing from their dual problem and block coordinate descent procedure.  \citet{mazumder2012exact} also derive screening rules for the \glasso~by arguing from the \glasso's optimality conditions.  Although all of these rules are \textit{safe}, \ie, they do not commit violations, we unfortunately do not use block coordinate descent to compute the \ours~estimate, and a careful inspection of \ours's optimality conditions reveals that these conditions are not separable in the entries of $\hat{\Omega}^\ourss$, making the framework of \citet{tibshirani2012strong} more appropriate here.

We state our rules in Lemma \ref{thm:screen}, and provide an algorithmic specification in Algorithm \ref{alg:screen}.
 
\begin{lemma}[Screening rules]
\label{thm:screen}
Let $\lambda_1^{(1)} \geq \cdots \geq \lambda_1^{(r-1)} \geq \lambda_1^{(r)}$ and $\lambda_2^{(1)} \geq \cdots \geq \lambda_2^{(s-1)} \geq \lambda_2^{(s)}$ form sequences of decreasing tuning parameters.  Also, let $\hat{\Omega}^\ourss(\lambda_1^{(k-1)}, \lambda_2^{(\ell)})$ be the solution of the \ours~optimization problem \eqref{eq:ours}, for a particular $\lambda_1^{(k-1)}$ and $\lambda_2^{(\ell)}$, $k \in \{2,\ldots,r\}$, $\ell \in \{1,\ldots,s\}$.  Finally, write the components of the gradient of the smooth parts of the objective in \eqref{eq:ours} evaluated at $\hat{\Omega}^\ourss(\lambda_1^{(k-1)}, \lambda_2^{(\ell)})$ as
\begin{align*}
c_{ij}(\lambda_1^{(k-1)}, \lambda_2^{(\ell)}) & = ( S_{ii} + S_{jj} + \lambda_2 ) \hat{\Omega}_{ij}^\ourss(\lambda_1^{(k-1)}, \lambda_2^{(\ell)}) + \sum_{j' \neq j}^p \hat{\Omega}_{i j'}^\ourss(\lambda_1^{(k-1)}, \lambda_2^{(\ell)}) S_{j j'} \\
&\quad + \sum_{i' \neq i}^p \hat{\Omega}_{i' j}^\ourss(\lambda_1^{(k-1)}, \lambda_2^{(\ell)}) S_{i i'}, \quad i,j=1,\ldots,p, \; i \neq j.
\end{align*}
Now, assume the $c_{ij}$ here are nonexpansive, \ie,
\[
\left| c_{ij}(\lambda_1^{(k)}, \lambda_2^{(\ell)}) - c_{ij}(\lambda_1^{(k-1)}, \lambda_2^{(\ell)}) \right| \leq \left| \lambda_1^{(k)} - \lambda_1^{(k-1)} \right|.
\]
Then we have that
\begin{equation}
\left| c_{ij}(\lambda_1^{(k-1)}, \lambda_2^{(\ell)}) \right| < 2 \lambda_1^{(k)} - \lambda_1^{(k-1)} \label{eq:screen}
\end{equation}
implies that $\hat{\Omega}_{ij}^\ourss(\lambda_1^{(k)}, \lambda_2^{(\ell)}) = 0$; \ie, the entries satisfying this condition can be omitted from the \ours~optimization problem \eqref{eq:ours} for $\lambda_1^{(k)}$ and $\lambda_2^{(\ell)}$.
\end{lemma}

\begin{algorithm}
\caption{Sequential strong screening rules for \ours}
\label{alg:screen}
\begin{algorithmic}
	\STATE \textbf{Input:} data matrix $X \in \reals^{n \times p}$; sequences of decreasing tuning parameters $(\lambda_1^{(k)})_{k=1}^{r}, (\lambda_2^{(\ell)})_{\ell=1}^{s}$
	\STATE \textbf{Output:} estimates $\hat{\Omega}^\ourss(\lambda_1^{(k)}, \lambda_2^{(\ell)}), \; k=1,\ldots,r, \; \ell=1,\ldots,s$

	\FOR{$\ell=1,\ldots,s$}
		\STATE compute $\hat \Omega^\ourss(\lambda_1^{(1)}, \lambda_2^{(\ell)})$ by solving Equation \ref{eq:ours} with $\lambda_1^{(1)}, \lambda_2^{(\ell)}$

		\FOR{$k=2,\ldots,r$}
		\STATE compute $N$, the set of nonzero variables, using Equation \ref{eq:screen} with $\hat \Omega^\ourss(\lambda_1^{(k-1)}, \lambda_2^{(\ell)}), \lambda_1^{(k-1)}, \lambda_2^{(\ell)}$
		\REPEAT
			\STATE compute $\hat \Omega^\ourss(\lambda_1^{(k)}, \lambda_2^{(\ell)})$ by solving Equation \ref{eq:ours} with $N, \lambda_1^{(k)}, \lambda_2^{(\ell)}$
			\STATE check (all variables) for violations using the optimality conditions for \eqref{eq:ours} (see Equation \ref{eq:kkt} in the supplement)
			\STATE add any violating variables back into $N$
		\UNTIL{there are no violations}
		\STATE output $\hat{\Omega}^\ourss(\lambda_1^{(k)}, \lambda_2^{(\ell)})$
		\ENDFOR
	\ENDFOR
\end{algorithmic}
\end{algorithm}

%% file: synth.tex
\section{Numerical examples}
\label{sec:exps}

We evaluate \ours, as well as several baselines, on synthetic and real-world data.  We are interested here not only in a method's variable selection accuracy, but also in its estimation error, which is a good measure of the method's suitability in a downstream application.  Previewing our findings a little, we see in our synthetic examples that \ours~significantly outperforms the (closely related) \conc~estimator in terms of estimation error, as measured by several matrix norms, while also outperforming \conc~in terms of variable selection accuracy --- these advantages also help \ours~outperform a number of strong baselines, when used in a real-world (non-Gaussian) finance application later on.  Finally, we see in a real-world sustainable energy example that the \ours~estimate is readily interpreted in a meaningful way; we also highlight the benefits of \ours's screening rules here.

\subsection{Synthetic data}
\label{sec:synth}

We begin by discussing the synthetic examples; in these, we directly compare to \conc, which is the method most closely related to ours.  We generated synthetic data as follows.  First, we generated a random, sparse, diagonally dominant $p \times p$ (ground truth) matrix $\Omegatrue$, by following the procedure in \citet{NIPS2014_5576, khare2014convex, peng2009partial,alnur}; we investigated $p \in \{ 1000, 3000 \}$.\footnote{This corresponds to estimating $p(p+1)/2 = 500,500$ and $4,501,500$ parameters, respectively.}  Then, we drew $n$ samples from a multivariate normal distribution with mean zero and covariance matrix $(\Omegatrue)^{-1}$, which were subsequently input into \ours~and \conc; we investigated $n \in \{0.2 p, 0.4 p, 0.8 p \}$ and $\lambda_1, \lambda_2 \in \{ 2^{-10}, 2^{-9.5}, \ldots, 1, 2^{0.5} \}$, \ie, a $22 \times 22$ grid.\footnote{Our experimental settings correspond to ultimately running \ours~and \conc~145,200 and 6,600 times, respectively.}  Finally, we computed the false and true positive rates for \ours~and \conc, by counting the number of nonzero entries in a method's estimate $\hat \Omega$ that were zero and nonzero, respectively, in $\Omegatrue$; we also computed the estimation error, \ie, $\| \Omegatrue - \hat{\Omega} \|$, in several matrix norms.  To summarize the variable selection accuracy and estimation errors across $\lambda_1, \lambda_2$, we computed the area under the curve (AUC), following, for example, \citet{NIPS2014_5576, khare2014convex,alnur}; to summarize the estimation errors, we computed the median across $\lambda_1, \lambda_2$.\footnote{Computing the mean across $\lambda_1, \lambda_2$ gave similar results.}  We repeated this entire process 50 times; thus, Tables \ref{tab:norms:onek} and \ref{tab:norms:threek} report the medians and interquartile ranges (IQRs) across these 50 trials.

Here, \ours~outperforms \conc~in AUC and estimation error across all sample sizes and norms (as well as on each trial individually).  \ours's estimation error, in particular, is significantly lower than \conc's; additionally, \ours's wallclock times as well as most of its interquartile ranges (IQRs) are generally lower than \conc's, demonstrating that the estimates produced by \ours~are quite stable.  These effects are likely due to the presence of the squared Frobenius norm penalty in the \ours~optimization problem.

\ifwantbiomtabs
\input{tab_norms_onek}

\input{tab_norms_threek}
\else
\input{tab_norms_onek_arxiv}

\input{tab_norms_threek_arxiv}
\fi

We also investigate the efficacy of \ours's screening rules; using the same synthetic data, we measure the (median across 50 trials) percentages of variables that the rules suggest dropping (excluding diagonal entries), as well as the percentages of violations (for $\lambda_2 = 1$).  Figure \ref{fig:screen} presents the results: the rules drop more variables as $\lambda_1$ increases (as expected), but never commit any violations.

\begin{figure}
\includegraphics[width=\fsw]{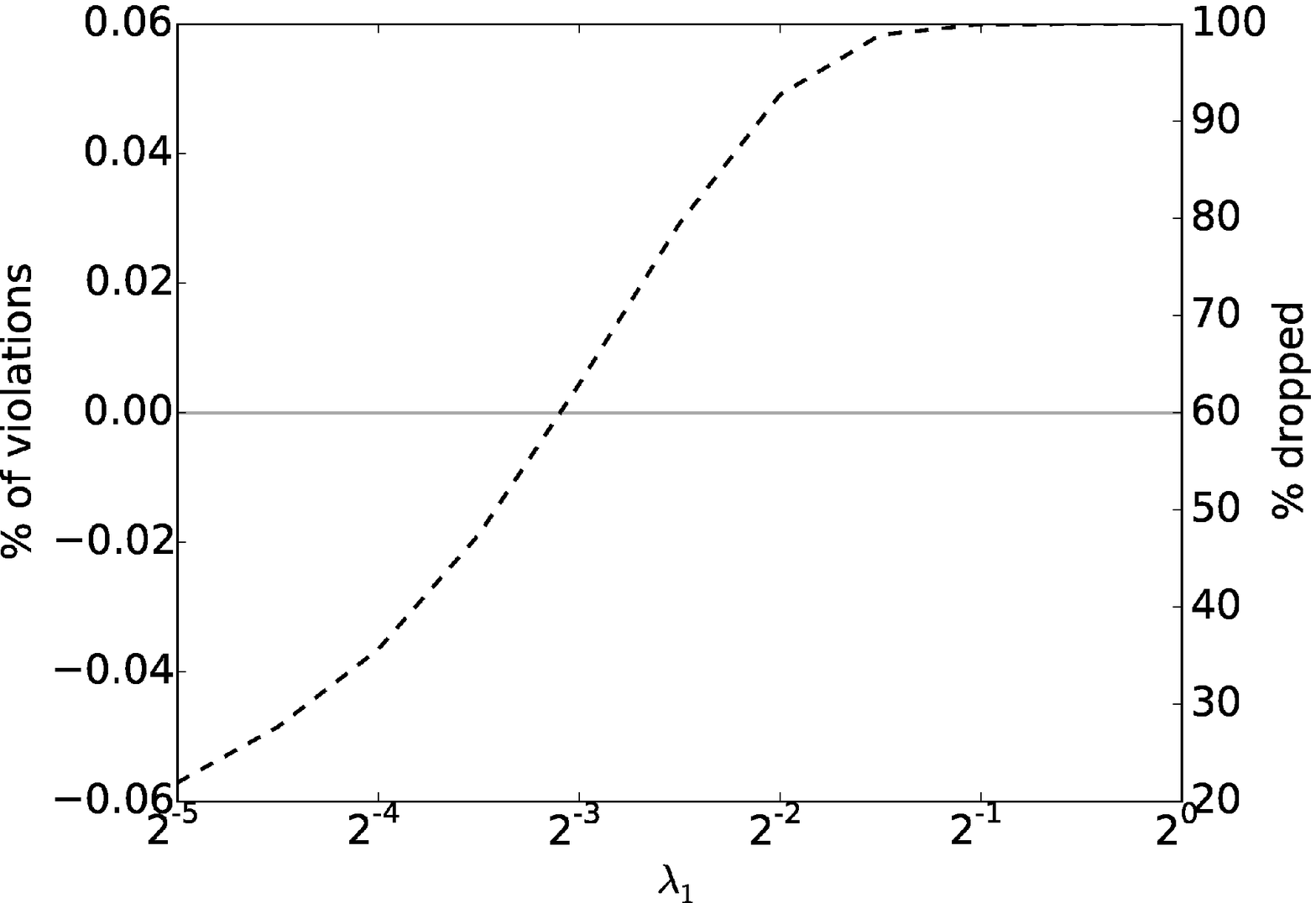}
\includegraphics[width=\fsw]{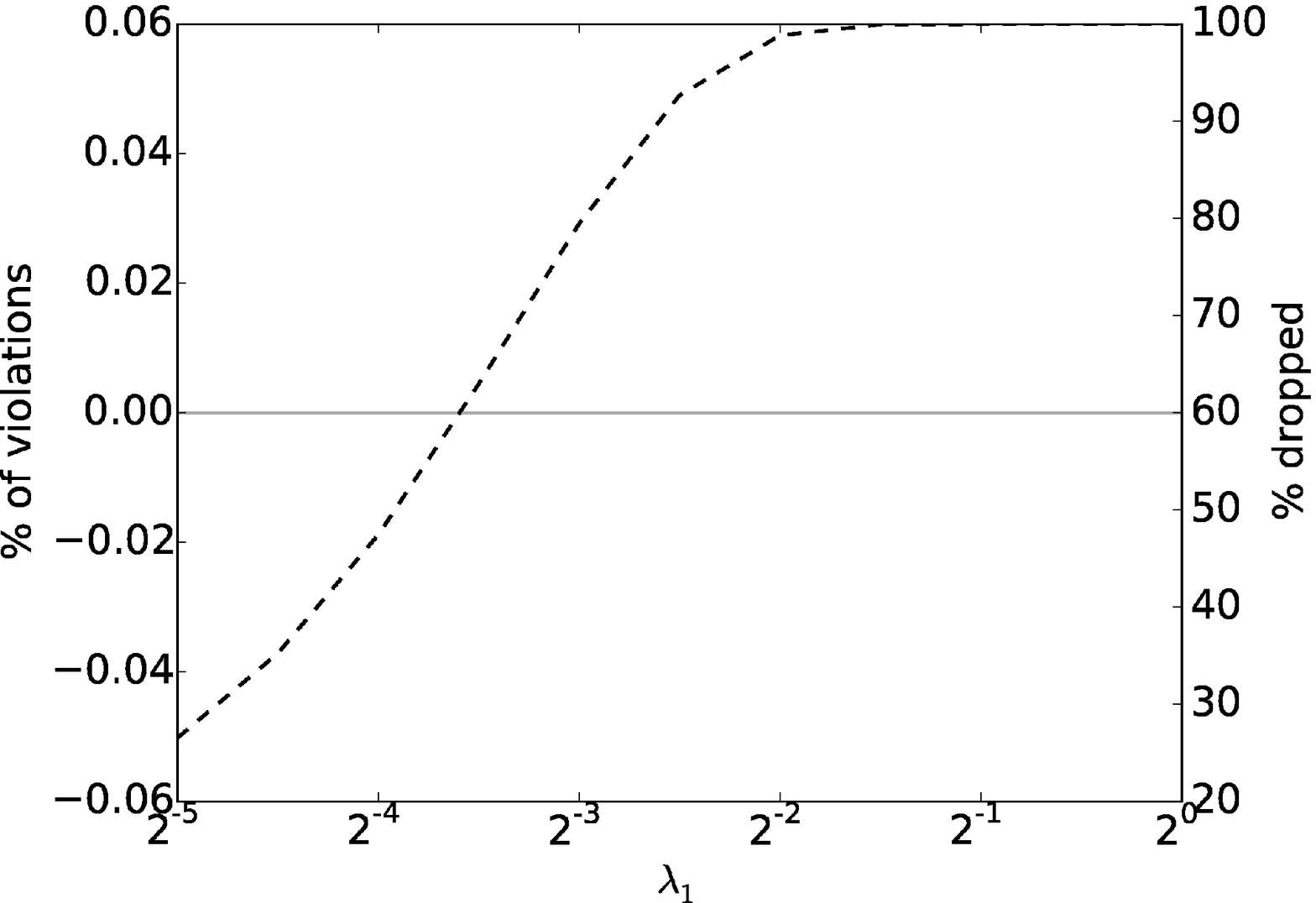}
\includegraphics[width=\fsw]{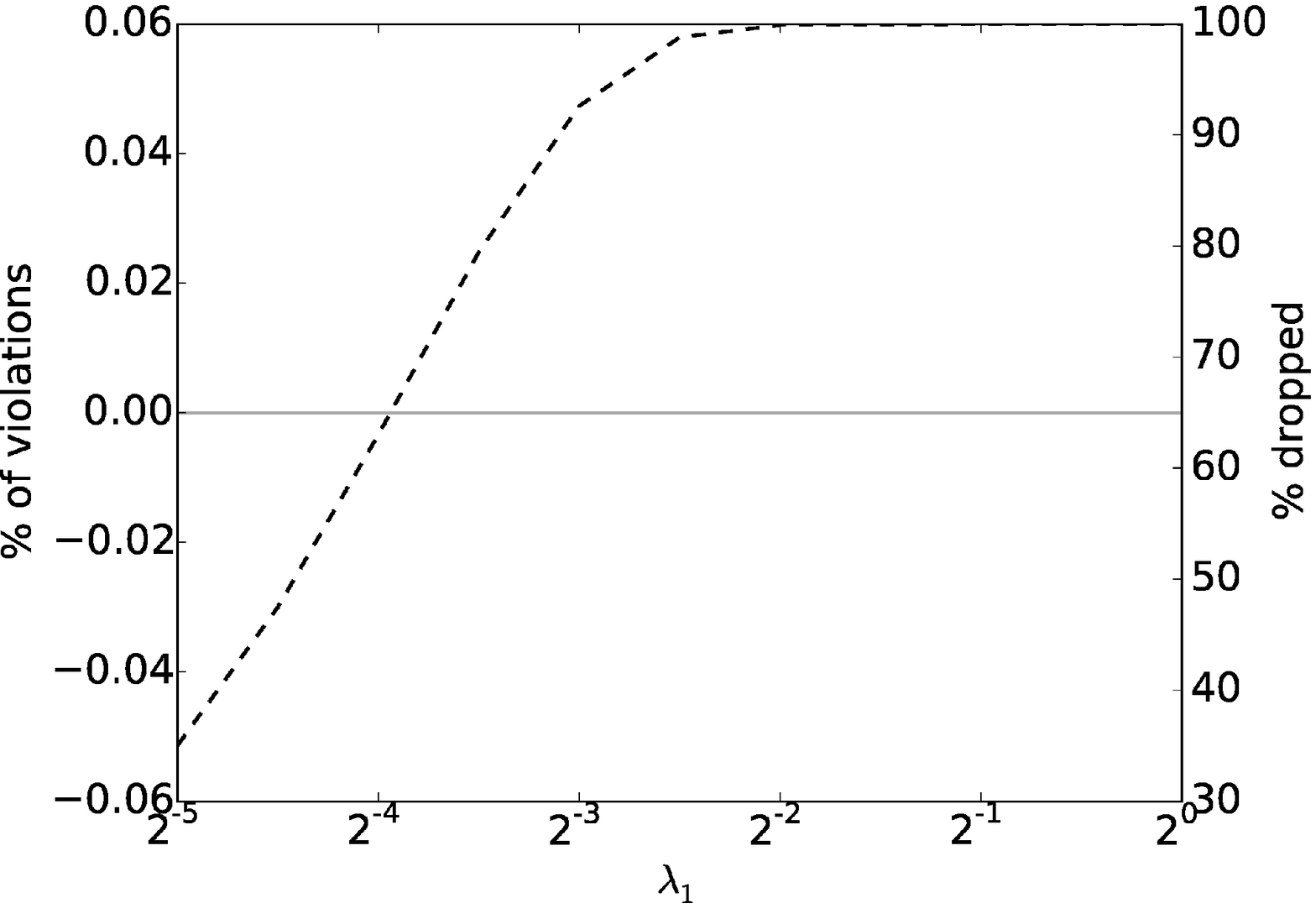} \\
\includegraphics[width=\fsw]{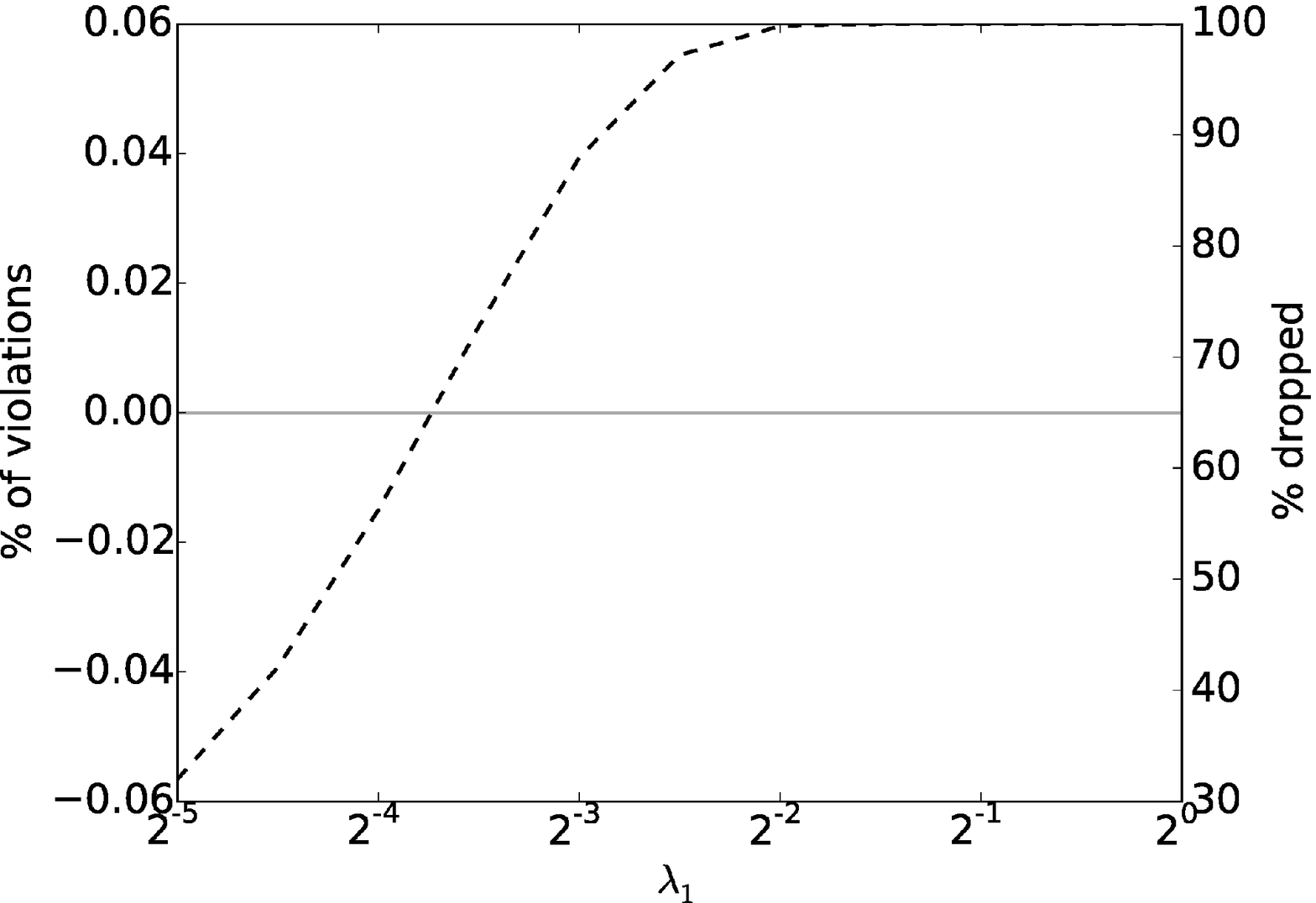}
\includegraphics[width=\fsw]{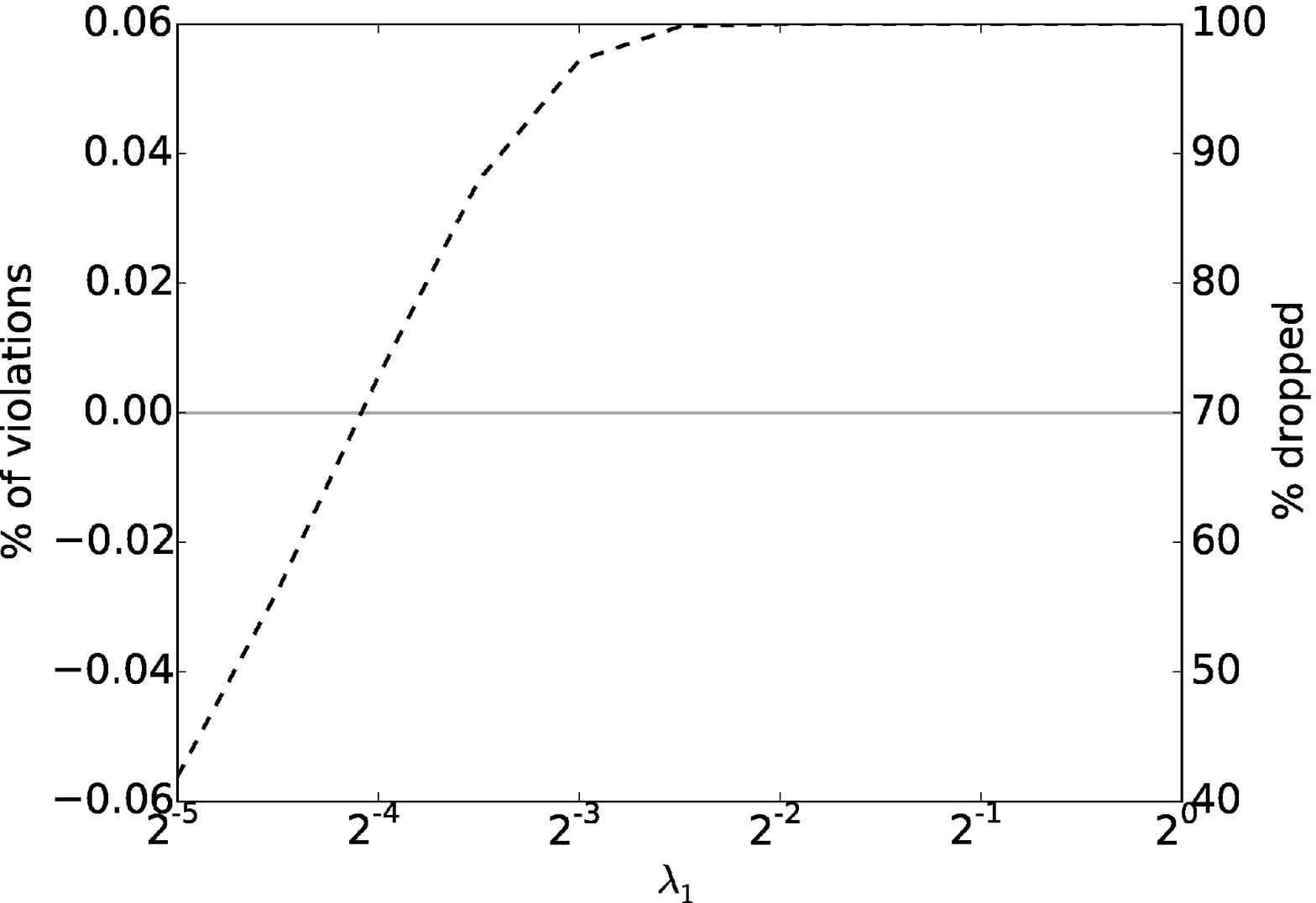}
\includegraphics[width=\fsw]{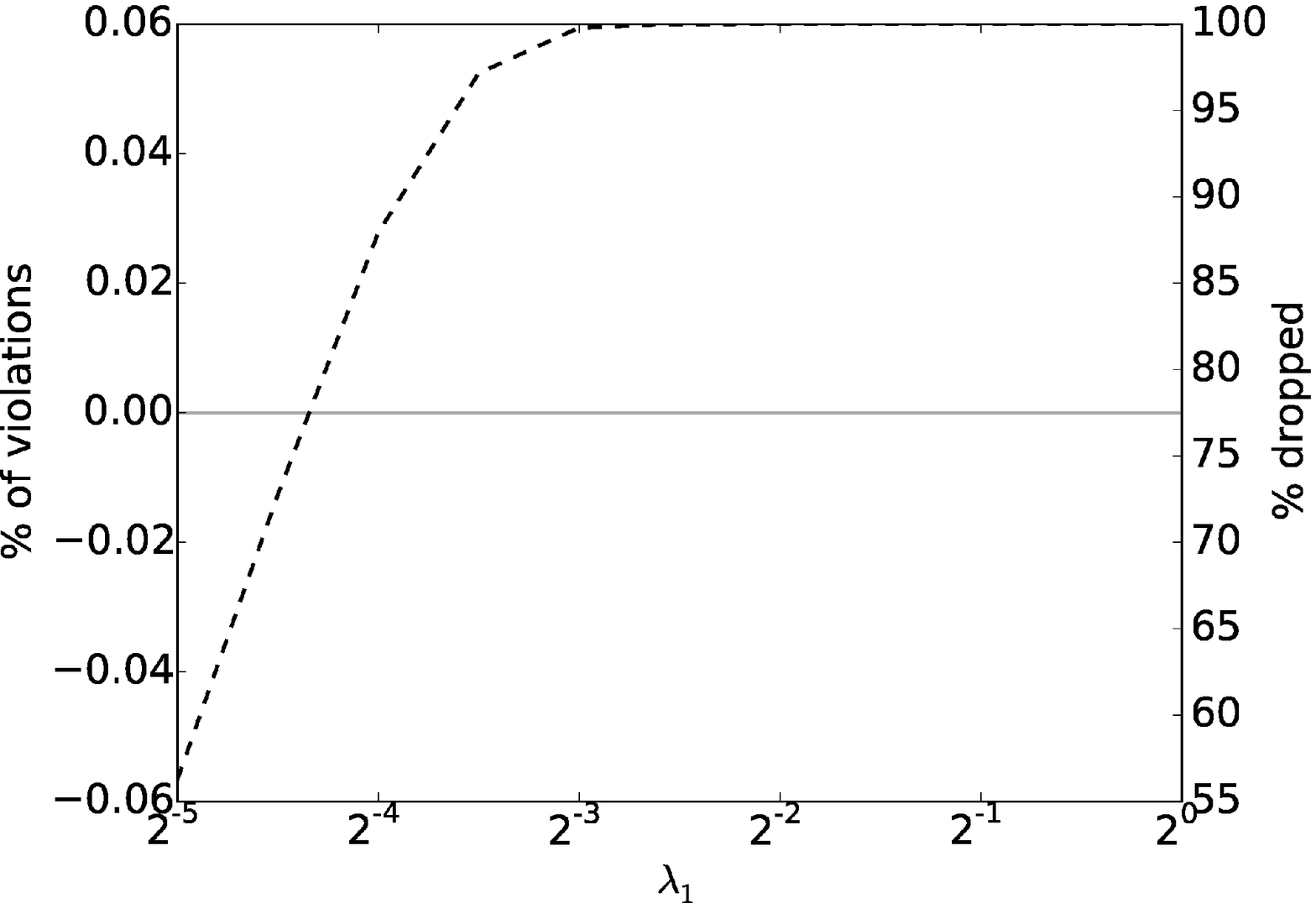}
\caption{Percentages of dropped variables excluding diagonal entries (dashed line, right vertical axes) and violations (solid line, left vertical axes) for \ours's screening rules ($\lambda_2 = 1$).  First row is $p = 1000$, second row is $p = 3000$; first column is $n = 0.2 p$, second is $n = 0.4 p$, third is $n = 0.8 p$.  The rules never commit a violation.}
\label{fig:screen}
\end{figure}

%% file: tab_norms_onek_arxiv.tex
\begin{table}
	\centering
	\begin{scriptsize}
	\begin{tabular}{llccccccc}
		\\
		& & \multicolumn{2}{c}{$n = 200$} & \multicolumn{2}{c}{$n = 400$} & \multicolumn{2}{c}{$n = 800$} \\
		& & \ours & \conc & \ours & \conc & \ours & \conc \\[5pt]
		\input{results_for_paper/synth/estim/p_1000}
	\end{tabular}
	\end{scriptsize}
	\caption{Median and interquartile range for \ours \;and \conc's areas under the curves (AUCs), estimation errors in several matrix norms, and wallclock times ($p = 1000$).  Higher median AUC is better, lower median estimation error and wallclock time is better; best in \textbf{bold}.  \ours \;outperforms \conc \;across all sample sizes and metrics.}
	\label{tab:norms:onek}
\end{table}

%% file: results_for_paper/synth/estim/p_1000.tex
\multirow{2}{*}{AUC} & Median & \textbf{0.68} & 0.65 & \textbf{0.81} & 0.73 & \textbf{0.91} & 0.86 \\
 & IQR & 0.02 & 0.01 & 0.01 & 0.01 & 0.01 & 0.01 \\
\multirow{2}{*}{Squared Frobenius norm} & Median & \textbf{6391.48} & 20150.68 & \textbf{5722.84} & 18805.59 & \textbf{4205.49} & 14990.35 \\
 & IQR & 84.70 & 513.99 & 26.18 & 245.65 & 18.22 & 192.78 \\
\multirow{2}{*}{$\ell_2$ operator norm} & Median & \textbf{2.51} & 5.17 & \textbf{2.41} & 5.07 & \textbf{2.56} & 5.84 \\
 & IQR & 0.01 & 0.06 & 0.01 & 0.03 & 0.01 & 0.03 \\
\multirow{2}{*}{Elementwise $\ell_1$ norm} & Median & \textbf{17480.45} & 35959.79 & \textbf{21640.10} & 46951.74 & \textbf{21749.16} & 51526.46 \\
 & IQR & 65.71 & 323.46 & 35.01 & 240.09 & 26.38 & 276.32 \\
\multirow{2}{*}{Elementwise $\ell_\infty$ norm} & Median & \textbf{1.34} & 2.93 & \textbf{1.06} & 2.32 & \textbf{0.67} & 1.38 \\
 & IQR & 0.01 & 0.04 & 0.01 & 0.02 & 0.01 & 0.03 \\
\multirow{2}{*}{Wallclock time (secs.)} & Median & \textbf{73.72} & 103.23 & \textbf{40.76} & 71.02 & \textbf{14.60} & 20.46 \\
 & IQR & 3.23 & 41.53 & 1.76 & 29.54 & 0.70 & 7.08

%% file: tab_norms_threek_arxiv.tex
\begin{table}
	\centering
	\begin{scriptsize}
	\begin{tabular}{llccccccc}
		\\
		& & \multicolumn{2}{c}{$n = 600$} & \multicolumn{2}{c}{$n = 1200$} & \multicolumn{2}{c}{$n = 2400$} \\
		& & \ours & \conc & \ours & \conc & \ours & \conc \\[5pt]
		\input{results_for_paper/synth/estim/p_3000}
	\end{tabular}
	\end{scriptsize}
	\caption{Median and interquartile range for \ours \;and \conc's areas under the curves (AUCs), estimation errors in several matrix norms, and wallclock times ($p = 3000$).  Higher median AUC is better, lower median estimation error and wallclock time is better; best in \textbf{bold}.  \ours \;outperforms \conc \;across all sample sizes and metrics.}
	\label{tab:norms:threek}
\end{table}

%% file: results_for_paper/synth/estim/p_3000.tex
\multirow{2}{*}{AUC} & Median & \textbf{0.64} & 0.63 & \textbf{0.75} & 0.71 & \textbf{0.86} & 0.84 \\
 & IQR & 0.01 & 0.01 & 0.00 & 0.01 & 0.01 & 0.01 \\
\multirow{2}{*}{Squared Frobenius norm} & Median & \textbf{15495.27} & 49063.26 & \textbf{12913.39} & 42021.80 & \textbf{8639.99} & 30054.52 \\
 & IQR & 83.60 & 75.39 & 4.46 & 78.99 & 21.98 & 34.91 \\
\multirow{2}{*}{$\ell_2$ operator norm} & Median & \textbf{2.17} & 4.48 & \textbf{2.01} & 4.19 & \textbf{1.99} & 4.43 \\
 & IQR & 0.00 & 0.00 & 0.00 & 0.01 & 0.00 & 0.00 \\
\multirow{2}{*}{Elementwise $\ell_1$ norm} & Median & \textbf{72178.79} & 148152.12 & \textbf{87484.12} & 187895.91 & \textbf{84109.25} & 195442.01 \\
 & IQR & 114.88 & 89.36 & 28.19 & 150.31 & 66.62 & 112.74 \\
\multirow{2}{*}{Elementwise $\ell_\infty$ norm} & Median & \textbf{1.10} & 2.38 & \textbf{0.83} & 1.77 & \textbf{0.49} & 0.95 \\
 & IQR & 0.00 & 0.01 & 0.00 & 0.01 & 0.00 & 0.01 \\
\multirow{2}{*}{Wallclock time (secs.)} & Median & \textbf{1861.35} & 3657.65 & \textbf{580.11} & 1208.06 & \textbf{124.72} & 236.40 \\
 & IQR & 7.86 & 36.14 & 1.48 & 7.43 & 0.06 & 2.14

%% file: finance.tex
\subsection{Minimum variance portfolio optimization}

Next, we evaluate \ours, as well as several other methods, in the context of a finance application.  We consider the problem of \textit{minimum variance portfolio optimization}, \ie, we must allocate our wealth across $p$ assets so that our overall risk is minimized; we model risk here as $x^T \hat{\Sigma} x$, where $x \in \reals^p$ is an allocation vector ($x_i > 0$ corresponds to a long position, while $x_i < 0$ corresponds to a short position), and $\hat{\Sigma}$ is an estimate of the underlying covariance matrix.  This leads to the following (convex) optimization problem: 
\optprobstartnn
\begin{array}{ll}
\minimizewrt{x \in {\bf R}^p} & x^T \hat{\Sigma} x \\
\subjectto & \ones^T x = 1,
\end{array}
\label{eq:finance}
\optprobendnn
which admits the analytical solution $x = ( \ones^T \hat{\Sigma}^{-1} \ones )^{-1} \hat{\Sigma}^{-1} \ones$.  We choose to solve a minimum variance portfolio optimization problem (instead of, say, a \textit{mean/variance problem} \citep{markowitz1952portfolio}) in order to isolate the impact of the estimate $\hat{\Omega} = \hat{\Sigma}^{-1}$.

\newcommand{\Nest}{H} 
We obtained the closing prices of the 30 constituent stocks of the Dow Jones Industrial Average (DJIA) from February 18, 1995 through October 26, 2012 (roughly 17 years) from \texttt{http://finance.yahoo.com}.  We divided the data into $T = 261$ consecutive time periods (of roughly 20 days each).  The $\Nest$ days preceding each trading period, commonly referred to as the \textit{estimation horizon}, were used to compute the estimate $\hat{\Omega}$; 10-fold cross-validation using the criterion \eqref{eq:bic} was used to choose $\lambda_1$ and $\lambda_2$.  The trading period was then used to evaluate the methods.  We investigated $\Nest \in \{ 35, 40, 45, 50, 75, 150, 225, 300 \}$.

We primarily evaluated each method using \textit{realized risk}, \ie,
\[
r = \left( (1/T) \sum_{t=1}^T \left( x_t^T p_t - \bar{p} \right)^2 \right)^{1/2},
\]
where $x_t, p_t \in \reals^p$ are the portfolio allocation and price change vectors for period $t$, respectively, and $\bar{p}$ is the \textit{realized return}, \ie,
\[
\bar{p} = (1/T) \sum_{t=1}^T x_t^T p_t,
\]
as well as the (commonly used) \textit{Sharpe ratio}, \ie,
\[
\left( \bar{p} - p_\textrm{free} \right)/ r,
\]
where $p_\textrm{free}$ is the risk-free rate (we set $p_\textrm{free} = 5\%$); intuitively, realized risk measures the instability (\ie, riskiness) of a trading strategy, and the Sharpe ratio trades off the (risk-free rate adjusted) returns and risk.

We compared \ours~with \conc, the sample covariance matrix (denoted \samplee), the \glasso, the condition number-regularized inverse covariance matrix estimator of \citet{won2013condition} (\condreg), the Ledoit-Wolf estimator \citep{ledoit2003honey} (\ledoit), as well as the DJIA itself (\ie, an index fund).  Tables \ref{tab:risk} and \ref{tab:sharpe} present the results.  When the estimation horizon is small, \ie, when $H \in \{35,40,45,50,75\}$, \ours~achieves the lowest risk, which is a useful feature when markets fluctuate; \ours~is always within 4\% of the lowest risk when the estimation horizon is larger.  Additionally, \ours~achieves significantly lower risk than \conc~across all estimation horizons.  These reductions in risk also translate into better Sharpe ratios for \ours: \ours~achieves the highest Sharpe ratio four (out of eight) times, which is more than any other method.  When \ours~does not achieve the highest Sharpe ratio, it is usually within 5\% of the best Sharpe ratio.  We also plot the cumulative wealth (in \$) achieved by an estimator (for $H = 300$) in Figure \ref{fig:wealth}.  \ours~achieves the highest cumulative wealth despite not (directly) optimizing for returns (\$8.75 for \ours~versus \$8.72 for \conc) while incurring less risk: \ours~also preserves the most wealth during the 2008--2009 financial crisis (\$4.64 for \ours~versus \$4.43 for \conc~and \$4.23 for \condreg).  Further details are provided in the supplement.

\ifwantbiomtabs
\input{tab_rlzd_risk}

\input{tab_sharpe}
\else
\input{tab_rlzd_risk_arxiv}

\input{tab_sharpe_arxiv}
\fi

\begin{figure}
\includegraphics[width=\textwidth]{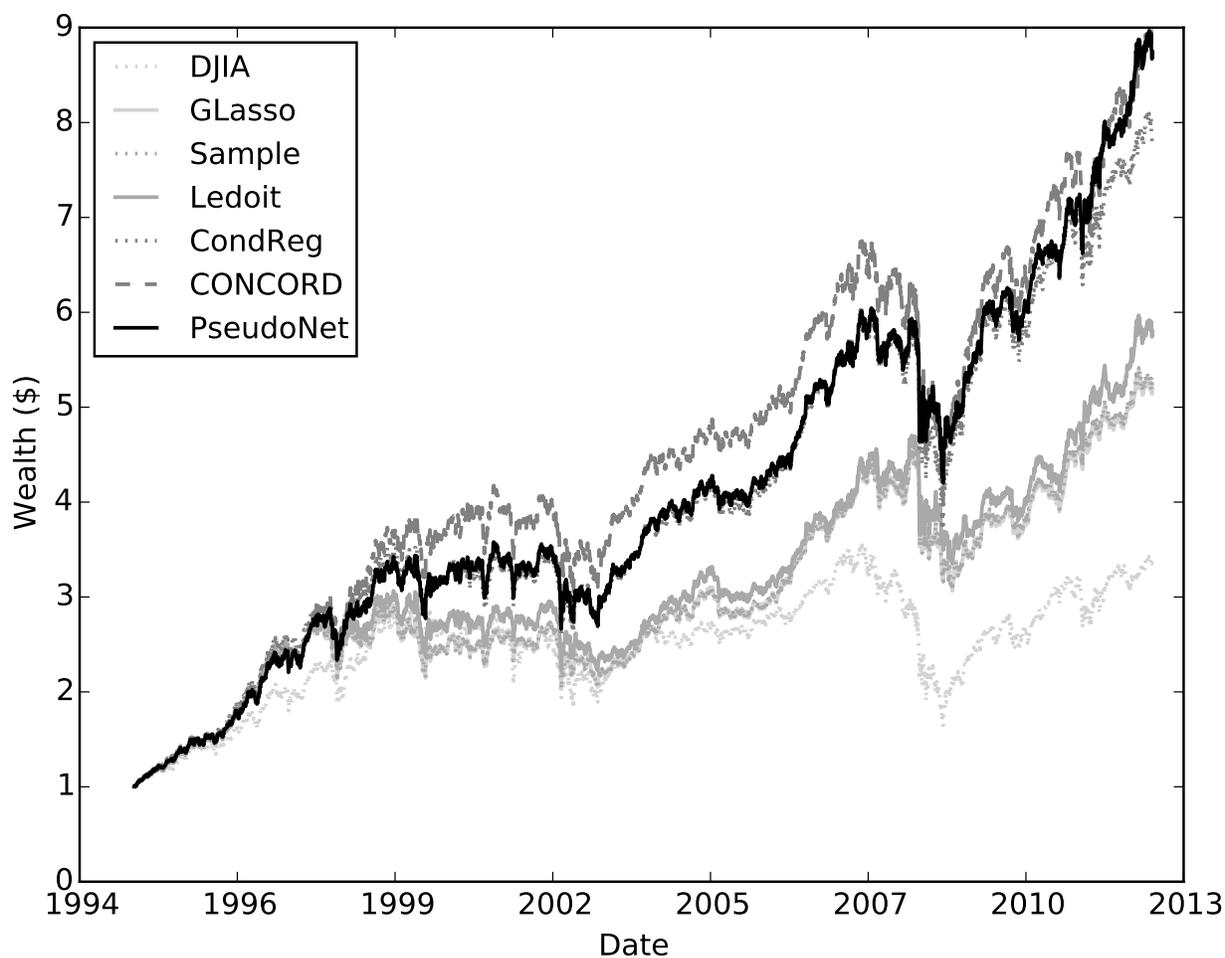}
\caption{Cumulative wealth for various estimators in the portfolio optimization example ($H = 300$); higher is better.  \ours~achieves the highest cumulative wealth.} 
\label{fig:wealth}
\end{figure}

%% file: tab_rlzd_risk_arxiv.tex
\begin{table}
	\centering
	\begin{scriptsize}
	\begin{tabular}{llccccccc}
		\\
			$H$ & \ours & \conc & \samplee & \glasso & \condreg & \ledoit & \djia \\[5pt]
			\input{results_for_paper/portopt/rlzd_risk}
	\end{tabular}
	\end{scriptsize}
	\caption{Realized risk for various estimators and estimation horizons $H$ in the portfolio optimization example.  Lower is better; best in \textbf{bold}.  \ours~is best 5/8 times.}
	\label{tab:risk}
\end{table}

%% file: results_for_paper/portopt/rlzd_risk.tex
35 & \textbf{15.23} & 17.03 & 33.86 & 16.55 & 17.83 & 15.58 & 18.96\\
40 & \textbf{15.04} & 17.02 & 26.52 & 16.54 & 17.76 & 15.46 & 18.96\\
45 & \textbf{15.21} & 17.04 & 23.19 & 16.56 & 17.64 & 15.43 & 18.96\\
50 & \textbf{15.01} & 17.02 & 20.95 & 16.36 & 17.61 & 15.36 & 18.96\\
75 & \textbf{15.06} & 17.04 & 17.45 & 15.61 & 17.20 & 15.10 & 18.96\\
150 & 15.07 & 17.09 & 15.41 & 14.99 & 16.37 & \textbf{14.66} & 18.96\\
225 & 15.12 & 17.10 & 14.98 & 14.87 & 16.07 & \textbf{14.52} & 18.96\\
300 & 15.25 & 17.16 & 14.95 & 14.95 & 16.10 & \textbf{14.52} & 18.96\\

%% file: tab_sharpe_arxiv.tex
\begin{table}
	\centering
	\begin{scriptsize}
	\begin{tabular}{llccccccc}
			$H$ & \ours & \conc & \samplee & \glasso & \condreg & \ledoit & \djia \\[5pt]
			\input{results_for_paper/portopt/sharpe}
	\end{tabular}
	\end{scriptsize}
	\caption{Sharpe ratios for various estimators and estimation horizons $H$ in the portfolio optimization example.  Higher is better; best in \textbf{bold}.  \ours~is best 4/8 times.}
	\label{tab:sharpe}
\end{table}

%% file: results_for_paper/portopt/sharpe.tex
35 & \textbf{0.52} & 0.48 & 0.36 & 0.49 & 0.48 & 0.47 & 0.19\\
40 & \textbf{0.50} & 0.48 & 0.44 & 0.49 & 0.48 & 0.44 & 0.19\\
45 & 0.43 & \textbf{0.47} & 0.26 & 0.47 & 0.45 & 0.39 & 0.19\\
50 & 0.47 & \textbf{0.49} & 0.23 & 0.47 & 0.46 & 0.41 & 0.19\\
75 & \textbf{0.48} & 0.47 & 0.38 & 0.42 & 0.46 & 0.37 & 0.19\\
150 & 0.47 & \textbf{0.48} & 0.29 & 0.36 & 0.47 & 0.38 & 0.19\\
225 & 0.50 & 0.50 & 0.37 & 0.36 & \textbf{0.52} & 0.42 & 0.19\\
300 & \textbf{0.55} & 0.50 & 0.36 & 0.36 & 0.49 & 0.41 & 0.19\\

%% file: neuro.tex
\subsection{Sustainable energy application}
\label{sec:neuro}

Finally, we evaluate \ours~on the task of recovering the conditional independencies between several wind farms on the basis of historical wind power measurements at these farms; wind power is naturally intermittent (as are many renewable resources), and thus understanding the relationships between wind farms can help operators forecast, plan, and dispatch.  We obtained hourly wind power measurements from July 1, 2009 through September 14, 2010 (440 days) at seven wind farms from \texttt{http://www.kaggle.com/c/GEF2012-wind-forecasting}; see \citet{hong} for further details, as well as a summary of a recent Kaggle competition based on this data.  Each group of 48 columns in the data set corresponds to two days (\ie, 48 hours) of hourly wind power measurements at a particular farm; to model the nonlinear relationship between wind power at different locations, we consider five radial basis function kernels spread evenly and evaluated at each hourly measurement (see, for example, \citet{wytock2013sparse,alnur} for a similar approach).  Thus, $p = 7 \times 48 \times 5 = 1680$.  Each row in the data set considers wind power measurements starting 12 hours after the (start of the) previous row; for example, the first row considers wind power measurements from 1:00 pm on July 1, 2009 through 12:00 pm on July 3, 2009, the second row from 1:00 am on July 2, 2009 through 12:00 am on July 4, 2009, and the last row from 1:00 am on September 12, 2010 through 12:00 am on September 14, 2010.  Thus, $n = 877$.  Computing the \ours~estimate here therefore corresponds to learning the structure of a spatiotemporal graphical model.

The left panel of Figure \ref{fig:wind} presents the \ours~estimate's sparsity pattern.  The nonzero super- and sub-diagonal entries suggest that at any wind farm the previous hour's wind power (naturally) influences the next hour's, while the nonzero off-diagonal entries, for example, in the (4,6) block, uncover farms that may influence one another: for example, farms 4 and 6 may be nearby, or (perhaps more interestingly) they may not be nearby\footnote{The true wind farm locations are censored in the data set.}.  \citet{wytock2013sparse}, whose method placed fifth in the Kaggle competition, as well as \citet{alnur} report similar findings (see the left panel of Figure 7 as well as Figure S.3, respectively, in these papers).  The right panel of Figure \ref{fig:wind} evaluates \ours's screening rules on this data set: the rules never commit a violation.

\begin{figure}
\centering
\includegraphics[scale=0.4,valign=t]{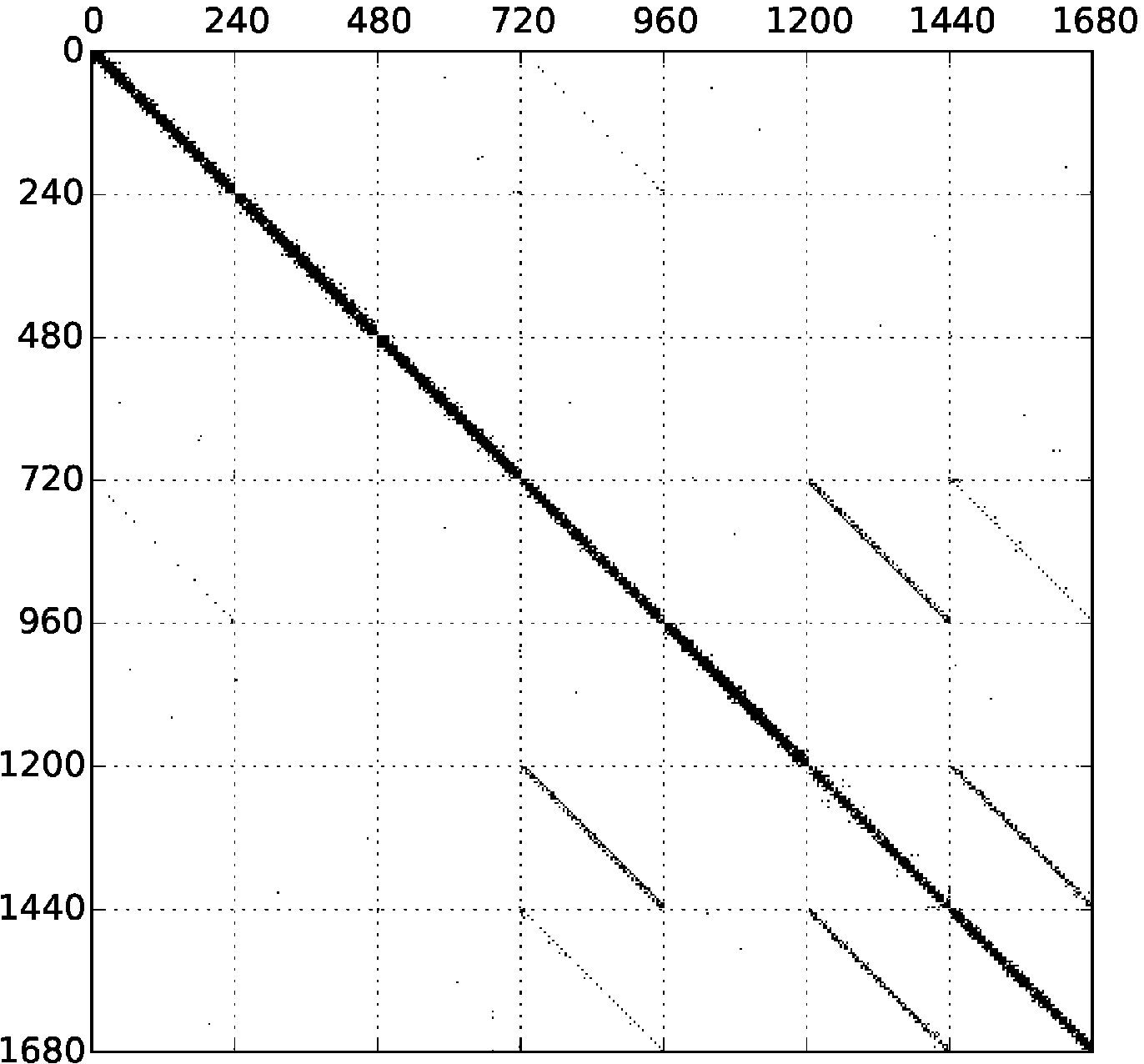}
\includegraphics[scale=0.4,valign=t]{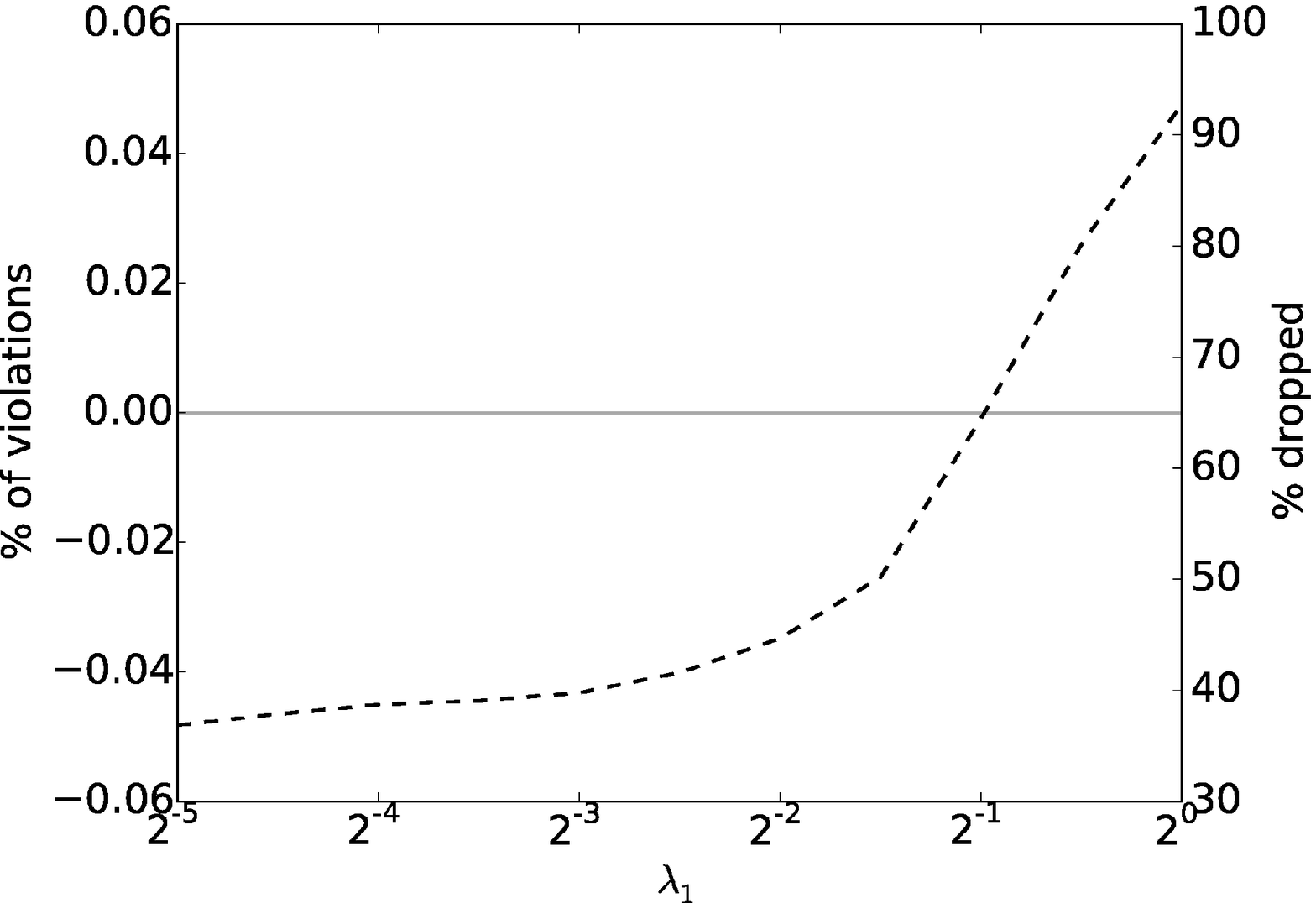}
\caption{Left: sparsity pattern for the \ours~estimate (black means nonzero, and each block corresponds to a wind farm).  Right: percentages of dropped variables excluding diagonal entries (dashed line, right vertical axes) and violations (solid line, left vertical axes) for \ours's screening rules ($\lambda_2 = 1$); the rules never commit a violation.}
\label{fig:wind}
\end{figure}

%% file: theory.tex
\section{Theory}
\label{sec:theory}

Finally, we collect here all our theoretical results on \ours's statistical and computational properties.  We state these results, essentially, in the order in which they are referenced in the text above.  Accordingly, we first show that the \ours~estimator converges to the unique, global solution of its defining optimization problem at a geometric (``linear'') rate.  Following this, we show, under suitable regularity conditions, that \ours~is consistent at a rate of $\sqrt{(\log p)/n}$; additionally, we provide a two-step method that obtains accurate estimates of the diagonal entries of the underlying inverse covariance matrix, even when $p > n$, as required by our consistency proof, which goes beyond the consistency proofs for the related pseudolikelihood-based estimators \spacee~\citep[Theorem 3]{peng2009partial} and \conc~\citep[Theorem 2]{khare2014convex}.  Finally, we show that the \ours~estimate does not saturate, while the \splice, \spacee, and \conc~estimates can saturate.  As a reminder, all proofs can be found in the supplement.

\subsection{Linear convergence}
\label{sec:conv}

We begin by showing that Algorithm \ref{alg:ours}, used to compute the \ours~estimate, converges to the unique, global solution of the \ours~optimization problem \eqref{eq:ours} at a geometric (``linear'') rate; this constrasts with a number of other pseudolikelihood-based methods, which do not provide unique estimates \citep{rocha2008path,peng2009partial,friedman2010applications,khare2014convex,NIPS2014_5576}, making interpretation difficult, are not guaranteed to converge \citep{rocha2008path,peng2009partial,friedman2010applications}, or converge at a slower rate \citep{khare2014convex,NIPS2014_5576}.

The result is given in Lemma \ref{thm:conv} below.

\begin{lemma}[Linear convergence]
\label{thm:conv}
Suppose $\Omega^{(0)}, \ldots, \Omega^{(k)}$ is a sequence of \ours~iterates with nonincreasing objective value.  Let $\hat{\Omega}^\ourss$ be the solution of the \ours~optimization problem \eqref{eq:ours}.  Then we get that
\begin{equation*}
\| \Omega^{(i)} - \hat{\Omega}^\ourss \|_F \leq (1-c)^i \| \Omega^{(0)} - \hat{\Omega}^\ourss \|_F, \quad i=1,\ldots,k,
\end{equation*}
where $c = \lambda_2 / L$ and $L$ is the Lipschitz constant for the gradient of the smooth term $\nabla g$ in \eqref{eq:prox2}.
\end{lemma}

\subsection{Consistency}
\label{sec:cons}

Next, we show, under suitable regularity conditions, that \ours~is consistent at a rate of $\sqrt{(\log p) / n}$.  Previous consistency results on pseudolikelihood-based estimators assume the existence of accurate estimates of the diagonal entries of the underlying inverse covariance matrix $\Omegatrue$; however, no method for obtaining such estimates is provided in these papers when $p > n$ \citep{khare2014convex,peng2009partial}.  Below, we provide a two-step method that obtains accurate diagonal estimates, which are required for the \ours~consistency proof (as well as for the consistency proofs for \conc~and \spacee); this is done in Theorem \ref{thm:diags}.

We now provide the regularity conditions required to establish the consistency of \ours; the assumptions are essentially the same as those required in \citet{khare2014convex}, which are in turn similar to those in \citet{peng2009partial}.
\begin{enumerate}[i.]
\item \textit{Sub-Gaussian rows.}  We require that the rows of the data matrix $X$ are i.i.d.~sub-Gaussian random vectors, \ie, there exists a constant $c \geq 0$ such that, for all $t \in \reals^p$, we have that $\Expect \exp ( t^T X_{i \cdot} ) \leq \exp( (c^2 / 2) t^T t ), \; i=1,\ldots,n$, where, as a reminder, $X_{i \cdot}$ is the $i$th row of $X$.

\item \textit{Correlation restrictions.}  For all $n$, we require that the minimum and maximum eigenvalues of the underlying covariance matrix $\Sigma^0$, \ie, $\lambda_{\min}(\Sigma^0)$ and $\lambda_{\max}(\Sigma^0)$, are uniformly bounded away from zero and infinity (note that we omit the notational dependence of $\Sigma^0$, as well as some related quantities, on $n$, for simplicity).

\item \textit{Incoherence.} We require that there exists a constant $\alpha < 1$ such that, for all $(i,j) \in \mathcal{A}_n^c$, where $\mathcal{A}_n$ here is the support of the off-diagonal entries of the underlying inverse covariance matrix $\Omegatrue_\X$, \ie,
\[
\mathcal{A}_n = \left\{ (i,j) : 1 \leq i < j \leq p, \; \Omegatrue_{ij} \neq 0 \right\},
\]
we have that
\begin{equation}
\left| \bar{L}_{ij, \mathcal{A}_n}^{''}(\omega^0_\X, \omega^0_\D) ( \bar{L}_{\mathcal{A}_n \mathcal{A}_n}^{''}(\omega^0_\X, \omega^0_\D) )^{-1} \sign \omega^0_{\mathcal{A}_n} \right| \leq \alpha.
\label{eq:incoh}
\end{equation}
Here, the $\sign$ here is interpreted elementwise; $\omega^0_\X$ and $\omega^0_\D$ are the vectorizations of the off-diagonal and diagonal entries, respectively, of the underlying inverse covariance matrix $\Omegatrue$, \ie,
\[
\omega^0_\X = \vect \Omegatrue_\X, \quad \omega^0_\D = \vect \Omegatrue_\D;
\]
$L(\omega_\X^0, \omega_\D^0)$ equals the $\log \det$ plus trace terms in \eqref{eq:ours} evaluated at $(\omega_\X^0, \omega_\D^0)$, \ie,
\[
L(\omega_\X^0, \omega_\D^0) = - (1/2) \log \det ( (\Omegatrue_\D)^2 ) + (n/2) \Tr ( S (\Omegatrue)^2 );
\]
and $\bar{L}^{''}_{ij,k\ell}(\omega_\X^0, \omega_\D^0)$ is an element of the negative $(p^2 \times p^2)$-dimensional Fisher information matrix at $(\omega_\X^0, \omega_\D^0)$, \ie,
\[
\bar{L}^{''}_{ij,k\ell}(\omega_\X^0, \omega_\D^0) = \Expect \frac{\partial^2 L(\omega_\X^0, \omega_\D^0)}{\partial \omega_{\X,ij}^0 \omega_{\X, k \ell}^0}, \quad i,j,k,\ell=1,\ldots,p
\]
(we abuse notation somewhat and write $\omega_{ij} = \Omega_{ij}$).

\item \textit{Accurate diagonal estimates.} We require the existence of accurate diagonal estimates $\hat{\omega}_\D$ such that
\[
\| \hat{\omega}_\D - \omega^0_\D \|_\infty = O_P( \sqrt{(\log n) / n} ).
\]
As stated in the beginning of this subsection, a method to obtain such estimates is provided in Theorem \ref{thm:diags}; our two-step method firstly performs a lasso regression (with tuning parameter $\lambda_{1,n}$) of each diagonal element on the remaining variables to identify subsets of relevant variables, and secondly estimates each diagonal element with the variance of the residuals given by the linear regression of each diagonal element on its subset of relevant variables.

\item \textit{Support size and tuning parameter restrictions.}  As $n \to \infty$, we let $q_n = o( \sqrt{n / \log n} )$, $\lambda_{1,n} \sqrt{q_n} \to 0$, $\lambda_{1,n} \sqrt{n / \log n} \to \infty$, and $\lambda_{2,n} = o(\lambda_{1,n})$, where $q_n = | \mathcal{A}_n |$ (note that we make explicit here the notational dependence of the tuning parameters on $n$).

\item \textit{Signal restrictions.}  As $n \to \infty$, we require that $s_n / ( \lambda_{1,n} \sqrt{q_n} ) \to \infty$, where $s_n = \max_{(i,j) \in \mathcal{A}_n} | \omega^0_{\X, ij} |$.
\end{enumerate}

Condition (iii) can be interpreted as requiring bounded correlation between the rows of $\bar{L}_{\mathcal{A}_n^c \mathcal{A}_n}^{''}(\omega^0_\X, \omega^0_\D)$ and the columns of $( \bar{L}_{\mathcal{A}_n \mathcal{A}_n}^{''}(\omega^0_\X, \omega^0_\D) )^{-1}$.  \citet{khare2014convex} as well as \citet{peng2009partial} also use this condition; see \citet{khare2014convex} for examples that satisfy this condition.

The following theorem presents our consistency result for \ours.

\begin{theorem}[Consistency]
\label{thm:cons}
Assume the conditions stated above.  Let $p = O(n^\kappa)$ for a constant $\kappa > 0$, and let $\hat{\Omega}^\ourss$ be the \ours~estimate given by the solution of the \ours~optimization problem \eqref{eq:ours}.  Then, we have, with probability at least $1-O(n^{-\beta})$ for a constant $\beta > 0$,
\begin{enumerate}[a.]
\item \textit{signed support recovery}: $\sign \hat{\omega}^\ourss_{ij} = \sign \omega^0_{ij}, \; i,j=1,\ldots,p$, where $\hat{\omega}^\ourss = \vect \hat{\Omega}^\ourss$ (we take $\sign 0 = 0$)
\item \textit{estimation error}: $\| \hat{\omega}^\ourss - \omega^0 \|_2 \leq c_1 \lambda_{1,n} \sqrt{q_n}$, for a constant $c_1 > 0$.
\end{enumerate}
\end{theorem}

\subsubsection{Accurate diagonal estimates}
\label{sec:diags}

The following theorem provides consistent estimates of the diagonal entries of the underlying inverse covariance matrix $\Omegatrue$.  In the case when $d_n$, which denotes the maximum number of nonzero entries in any row of $\Omegatrue$, is bounded in $n$, this theorem yields estimates satisfying condition (iv) above, even when $p > n$; this result is also useful in the context of consistency for \conc~\citep[Theorem 2]{khare2014convex} and \spacee~\citep[Theorem 3]{peng2009partial}, where such diagonal estimates are assumed, but a method to obtain them is not provided.

\begin{theorem}[Accurate diagonal estimates via two-step method]
\label{thm:diags}
Assume conditions (i), (ii), (v), and (vi) above.  Assume further that there exists a constant $\delta < 1$ such that
\begin{equation}
\left| \Sigma^0_{i, \mathcal{A}_n^j} \left( \Sigma^0_{\mathcal{A}_n^j, \mathcal{A}_n^j} \right)^{-1} \sign \Omegatrue_{\mathcal{A}_n^j, j} \right| \leq \delta, \quad i \notin \mathcal{A}_n^j, \; j=1,\ldots,p,
\label{eq:incoh2}
\end{equation}
where
\begin{align*}
d_n & = \max_{k=1,\ldots,p} \left| \left\{ \ell : \ell \in \{1,\ldots,p\}, \; \ell \neq k, \; \Omegatrue_{k \ell} \neq 0 \right\} \right|, \\
\mathcal{A}_n^j & = \left\{ k : k \in \{1,\ldots,p\}, \; k \neq j, \; \Omegatrue_{jk} \neq 0 \right\},
\end{align*}
and the $\sign$ in \eqref{eq:incoh2} is interpreted elementwise.  Now, for $j=1,\ldots,p$, let $\hat{\mathcal{A}}^j_n$ be the set of indices corresponding to the nonzero coefficients obtained by fitting a lasso regression of the $j$th diagonal element on the remaining variables (with tuning parameter $\lambda_{1,n}$).  Also, let $\hat{\omega}_{\D, j}$ be the sample variance of the $j$th diagonal element conditioned on the variables in $\hat{\mathcal{A}}^j_n$.  Then, for every $\beta > 0$, there exists a constant $c_2 > 0$ such that 
\[
\| \hat{\omega}_\D - \omega^0_\D \|_\infty \leq c_2 d_n \sqrt{(\log n) / n},
\]
with probability at least $1 - O(n^{-\beta})$.
\end{theorem}

We note that \eqref{eq:incoh2} is similar but not equivalent to condition (iii) above.

\subsection{Saturation}
\label{sec:sat}

Lastly, we show that the \ours~estimate does not saturate (\ie, when $p \gg n$, the number of variables selected by \ours~can be greater than $np$ out of $p(p-1)/2$ total variables), while the \splice, \spacee, and \conc~estimates can saturate; this is rather limiting for these latter estimators from the points of view of both estimation error as well as interpretability.

To do this, we first introduce some notation that makes the statements of these results, as well as their proofs, more concise.  We use $\vech$ to mean the \textit{half-vectorization} operator, \ie, the concatenation of the lower triangle of its (matrix) argument, excluding diagonal entries.  We use $\card$ to count the number of nonzero entries in its argument.  Also, we say that the columns of a wide matrix $A \in \reals^{k \times \ell}$ (\ie, $\ell > k$) are in \textit{general position} if the affine span of any $m \leq k$ signed columns of $A$, \ie, $s_{i_1} A_{i_1}, \ldots, s_{i_m} A_{i_m}$, where each $s_j, \; j=i_1,\ldots,i_m$ is fixed to either $+1$ or $-1$, does not contain any of the points $\pm A_j, \; j \neq i_1,\ldots i_m$.

Below, Theorem \ref{thm:sat} states our saturation results for \ours~and \conc; Corollary \ref{thm:sat2} then gives the analogous results for \splice~and \spacee.

\begin{theorem}[Saturation results for \ours~and \conc]
\label{thm:sat}
Let
\ifwantbiomtabs
\else
\begin{scriptsize}
\fi
\begin{equation*}
A =
- \mtxstart{cccccccccccccccccccccccc}
X_{2} & X_{3} & X_{4} & \cdots & X_{p-1} & X_{p} &
0 & \multicolumn{16}{c}{\cdots} & 0 \\
X_{1} & 0 & \multicolumn{3}{c}{\cdots} & 0 &
X_{3} & X_{4} & X_{5} & \cdots & X_{p-1} & X_{p} &
0 & \multicolumn{10}{c}{\cdots} & 0 \\
0 & X_{1} & 0 & \multicolumn{2}{c}{\cdots} & 0 & 
X_{2} & 0 & \multicolumn{3}{c}{\cdots} & 0 &
X_{4} & X_{5} & X_{6} & \cdots & X_{p-1} & X_{p} &
0 & \multicolumn{4}{c}{\cdots} & 0 \\
\multicolumn{24}{c}{\vdots} \\
0 & \multicolumn{3}{c}{\cdots} & 0 & X_{1} & 0 & 
\multicolumn{3}{c}{\cdots} & 0 & X_{2} & 0 & \multicolumn{3}{c}{\cdots} & 0 & X_{3} & 0 & 
\multicolumn{3}{c}{\cdots} & 0 & X_{p-1}
\mtxend,
\end{equation*}
\ifwantbiomtabs
\else
\end{scriptsize}
\fi
\ie, $A \in \reals^{np \times p(p-1)/2}$ is a matrix containing the columns of the data matrix $X$ arranged in a particular fashion.  Also, let $\hat{\Omega}^\ourss$ be the \ours~estimate, \ie, the solution of the \ours~optimization problem \eqref{eq:ours}, and let $\hat{\Omega}^\concs$ be a \conc~estimate; so, we have $\vech \hat{\Omega}^\ourss, \vech \hat{\Omega}^\concs \in \reals^{p(p-1)/2}$.  Assume that $p \gg n$.  Then, the \ours~estimate does not saturate, \ie, $\card \vech \hat{\Omega}^\ourss \leq p(p-1)/2$, and there exists a \conc~estimate that saturates, \ie, $\card \vech \hat{\Omega}^\concs \leq np$.  Furthermore, if the columns of the matrix $A$ are in general position, then all \conc~estimates saturate.
\end{theorem}

The analogous results for \splice~and \spacee~follow by using arguments similar to those given in the proof of Theorem \ref{thm:sat}; to make the statement of these results clearer, we first describe the \splice~and \spacee~estimators in more detail.

We can obtain a \splice~estimate by first minimizing the following objective, alternately over the variables $D \in \reals^{p \times p}$ and $B \in \reals^{p \times p}$, where $D$ is a diagonal matrix and the diagonal entries of the matrix $B$ are set to zero, 
\begin{equation}
-(1/2) \log \det D + (1/2) \sum_{i=1}^p \frac{1}{D_{ii}^2} \| X_i - X_{-i} (B_{i \cdot})^T \|_2^2 + \lambda_1 \| B \|_1,
\label{eq:splice}
\end{equation}
where $X_{-i}$ denotes the data matrix $X$ after removing the $i$th column, and $B_{i \cdot}$ here means the $i$th row of $B$ after removing the entry $B_{ii}$; then, for any iteration $i$, we compute the estimate
\begin{equation}
\hat \Omega^{\spls, (i)} = (\hat D^{(i-1)})^{-2} (I - \hat B^{(i)}),
\label{eq:splice2}
\end{equation}
with $\hat \Omega^{\spls, (i)}$ referring to the estimate at the end of the $i$th iteration ($\hat D^{(i-1)}$ and $\hat B^{(i)}$ are interpreted similarly).

Turning to \spacee, we can compute a \spacee~estimate by minimizing the following objective, alternately over the variables $\Omega_\D$ and $\Omega_\X$,
\begin{equation}
-(1/2) \log \det \Omega_\D + (1/2) \sum_{i=1}^p \Omega_{\D, ii} \left\| X_i - \sum_{j \neq i}^p \Omega_{\X, ij} \sqrt{\Omega_{\D, jj} / \Omega_{\D, \; ii}} X_j \right\|_2^2 + \lambda_1 \| \Omega_\X \|_1,
\label{eq:space}
\end{equation}
where $\Omega_{\D, ii}$ refers to the $(i,i)$th entry of $\Omega_\D$ ($\Omega_{\X, ij}$ is interpreted similarly).  As a reminder, $\Omega_\D \in \reals^{p \times p}$ is a matrix of the diagonal entries of $\Omega$, with its off-diagonal entries set to zero; $\Omega_\X \in \reals^{p \times p}$ is a matrix of the off-diagonal entries of $\Omega$, with its diagonal entries set to zero; and we form the \spacee~estimate, for any iteration $i$, as $\hat{\Omega}^{\spcs, (i)} = \hat \Omega_\D^{(i)} + \hat \Omega_\X^{(i)}$.  To be clear, the superscripts involving $i$ here are interpreted just as with \splice~above (also, we note that in the optimization problem \eqref{eq:space}, we have set the ``weights'' for each regression subproblem $i$ to $\Omega_{\D, ii}$, as recommended by \citet{peng2009partial}).

Corollary \ref{thm:sat2} below gives the corresponding results for \splice~and \spacee.

\begin{corollary}[Saturation results for \splice~and \spacee]
\label{thm:sat2}
Let $\hat{\Omega}^{\spls, (i)}$ be a \splice~estimate at the end of iteration $i$, \ie, a solution of the optimization problem \eqref{eq:splice} and Equation \ref{eq:splice2}, and let $\hat{\Omega}^{\spcs, (i)}$ be a \spacee~estimate at the end of iteration $i$, \ie, a solution of the optimization problem \eqref{eq:space}; so, we have $\vech \hat{\Omega}^{\spls, (i)}, \vech \hat{\Omega}^{\spcs, (i)} \in \reals^{p(p-1)/2}$.  Assume that $p \gg n$.  Then, there exist \splice~and \spacee~estimates at the end of iteration $i$ that saturate, \ie, $\card \vech \hat{\Omega}^{\spls, (i)} \leq np$ and $\card \vech \hat{\Omega}^{\spcs, (i)} \leq np$.
\end{corollary}

%% file: disc.tex
\section{Discussion}
\label{sec:disc}

We introduced \ours, a new, more flexible pseudolikelihood-based estimator of the inverse covariance matrix; \ours~can be viewed as generalizing several Gaussian likelihood-based, as well as pseudolikelihood-based, estimators in ways that give \ours~a number of statistical and computational advantages.  We showed, through a number of experiments, that \ours~significantly outperforms the closely related \conc~estimator, in terms of both estimation error and variable selection accuracy, and that \ours~deals effectively with non-Gaussian data, making it well-suited for use in downstream applications.  We also showed, under regularity conditions, that \ours~is consistent at a rate of $\sqrt{(\log p)/n}$; our proof assumes the existence of accurate estimates of the diagonal entries of the underlying inverse covariance matrix (like \spacee~and \conc), and also provides a two-step method to obtain these estimates, even when $p > n$ (going beyond \spacee~and \conc).  Unlike several other pseudolikelihood-based methods, we also showed that the \ours~estimate does not saturate (\ie, when $p \gg n$, the number of variables selected by \ours~can be greater than $np$ out of $p(p-1)/2$ total variables), which is useful from both the perspectives of estimation error and interpretability.  We presented a fast algorithm for computing the \ours~estimate; we showed that this algorithm converges at a geometric (``linear'') rate to the unique, global solution of the \ours~optimization problem, and that it is faster than \conc.  Finally, we presented sequential strong screening rules that make computing the \ours~estimate over a range of tuning parameters much more tractable.  As a whole, we believe these statistical and computational properties represent a useful step forward in the design of pseudolikelihood-based estimators of the inverse covariance matrix.

%% file: supp_screen.tex
\section{Proof of Lemma \ref{thm:screen}}
\label{sec:screen}
\begin{proof}
By considering the gradient of the smooth term in the objective of the \ours~optimization problem \eqref{eq:ours}, given by \eqref{eq:grad}, in a componentwise fashion, we can express the optimality conditions for \eqref{eq:ours}, evaluated at the off-diagonal entries of $\hat{\Omega}_{ij}^\ourss(\lambda_1^{(k)}, \lambda_2^{(\ell)})$, as
\begin{equation}
\begin{array}{llll}
\left| c_{ij}(\lambda_1^{(k)}, \lambda_2^{(\ell)}) \right| & \leq \lambda_1^{(k)} & \quad \textrm{if} \quad \hat{\Omega}_{ij}^\ourss(\lambda_1^{(k)}, \lambda_2^{(\ell)}) & = 0 \\
c_{ij}(\lambda_1^{(k)}, \lambda_2^{(\ell)}) & = \lambda_1^{(k)} & \quad \textrm{if} \quad \hat{\Omega}_{ij}^\ourss(\lambda_1^{(k)}, \lambda_2^{(\ell)}) & > 0 \\
c_{ij}(\lambda_1^{(k)}, \lambda_2^{(\ell)}) & = -\lambda_1^{(k)} & \quad \textrm{if} \quad \hat{\Omega}_{ij}^\ourss(\lambda_1^{(k)}, \lambda_2^{(\ell)}) & < 0.
\end{array}
\label{eq:kkt}
\end{equation}

But, we have that
\begin{align*}
\left| c_{ij}(\lambda_1^{(k)}, \lambda_2^{(\ell)}) \right| & \leq \left| c_{ij}(\lambda_1^{(k)}, \lambda_2^{(\ell)}) - c_{ij}(\lambda_1^{(k-1)}, \lambda_2^{(\ell)}) \right| + \left| c_{ij}(\lambda_1^{(k-1)}, \lambda_2^{(\ell)}) \right| \\
& < | \lambda_1^{(k)} - \lambda_1^{(k-1)} | + 2 \lambda_1^{(k)} - \lambda_1^{(k-1)} \\
& = \lambda_1^{(k)},
\end{align*}
with the first inequality following by the triangle inequality, and the second by the assumptions that the $c_{ij}$ are nonexpansive and nonincreasing, as well as the further assumption that $\left| c_{ij}(\lambda_1^{(k-1)}, \lambda_2^{(\ell)}) \right| < 2 \lambda_1^{(k)} - \lambda_1^{(k-1)}$; by checking \eqref{eq:kkt}, this implies that $\hat{\Omega}_{ij}^\ourss(\lambda_1^{(k)}, \lambda_2^{(\ell)}) = 0$ is a solution.
\end{proof}

%% file: supp_finance.tex
\section{Additional numerical results for the minimum variance portfolio optimization example}
In addition to the numerical results given in the main paper, we consider here the realized risk and Sharpe ratios for various estimators and estimation horizons, after accounting for borrowing costs (at an 8\% annual percentage rate) and transaction costs (at 0.5\% of the principal); Tables \ref{tab:risk:adj} and \ref{tab:sharpe:adj} present the results, and we generally see the same trends as in the main paper.  \ours~achieves the lowest risk when the estimation horizon is small, and otherwise is within 5\% of the lowest risk.  \ours~also achieves the highest Sharpe ratio four (out of eight) times, and is otherwise within 5\% of the highest Sharpe ratio.

\ifwantbiomtabs
\input{tab_risk_adj}

\input{tab_sharpe_adj}
\else
\input{tab_risk_adj_arxiv}

\input{tab_sharpe_adj_arxiv}
\fi

Qualitatively, we find that, although \ours~does provide sparse estimates, these estimates are usually somewhat denser than those provided by \conc~(as expected); Figure \ref{fig:heatmap} plots these estimates (from a randomly chosen investment horizon and trading period).  Thus, owing to its (comparatively) denser and better estimates, \ours~can reduce risk by hedging, for example, by taking a short position in a stock whose returns are negatively correlated with another stock that it also takes a long position in.  To this end, we consider the \textit{size of the short side} of a portfolio $x \in \reals^p$, which is defined as the ratio of the magnitude of all the short positions in the portfolio to the magnitude of the portfolio, expressed as a percentage, \ie,
\[
100 \times \left( \sum_{i=1}^p \min \{ x_i, 0 \} \right) / \left( \sum_{i=1}^p | x_i | \right).
\]
Table \ref{tab:short} presents the size of the short side, averaged over all trading periods, for various estimators and estimation horizons, and we indeed see that the size of \ours's short side is larger than \conc's, \glasso's, and \condreg's.

\begin{figure}
\includegraphics[width=0.48\textwidth]{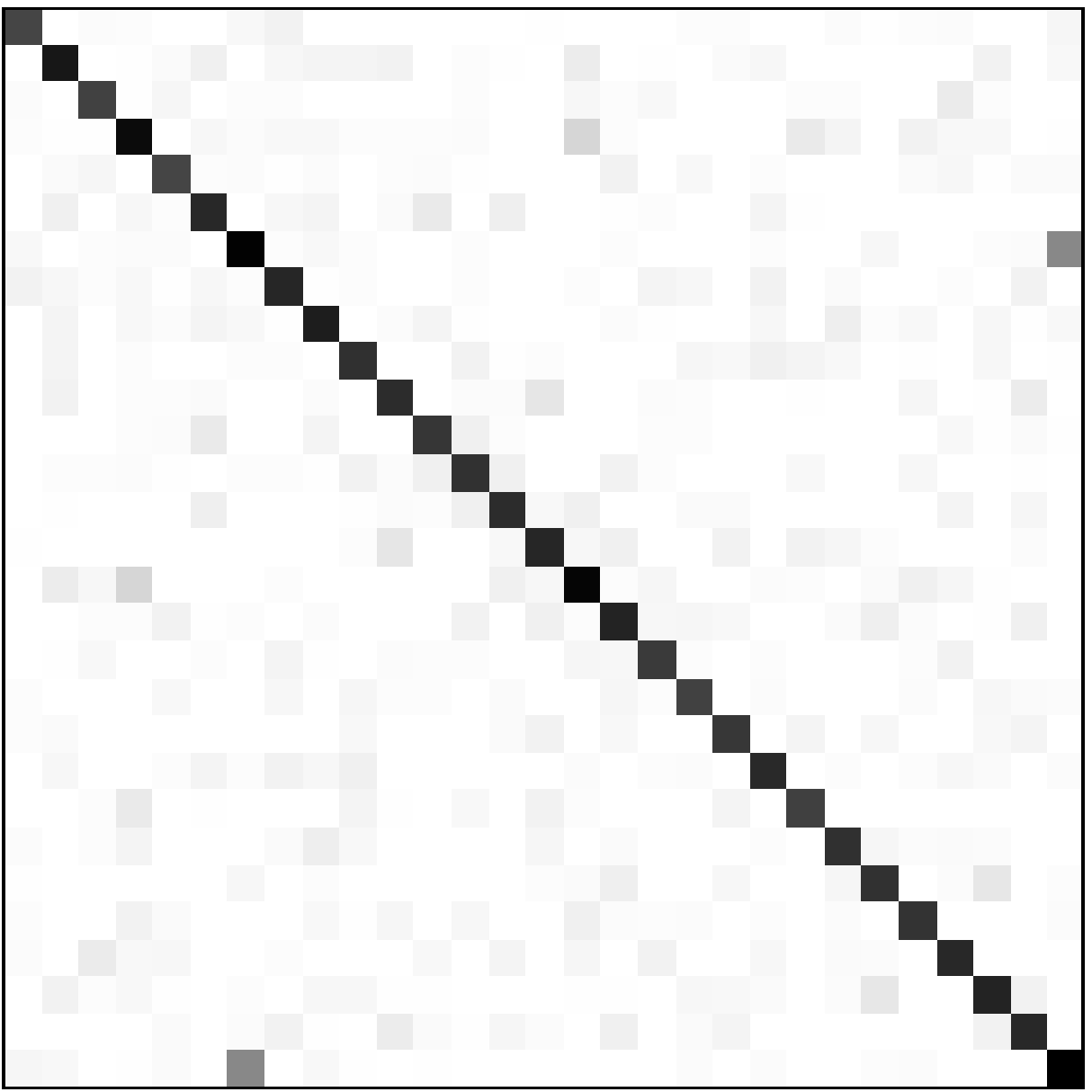}
\includegraphics[width=0.48\textwidth]{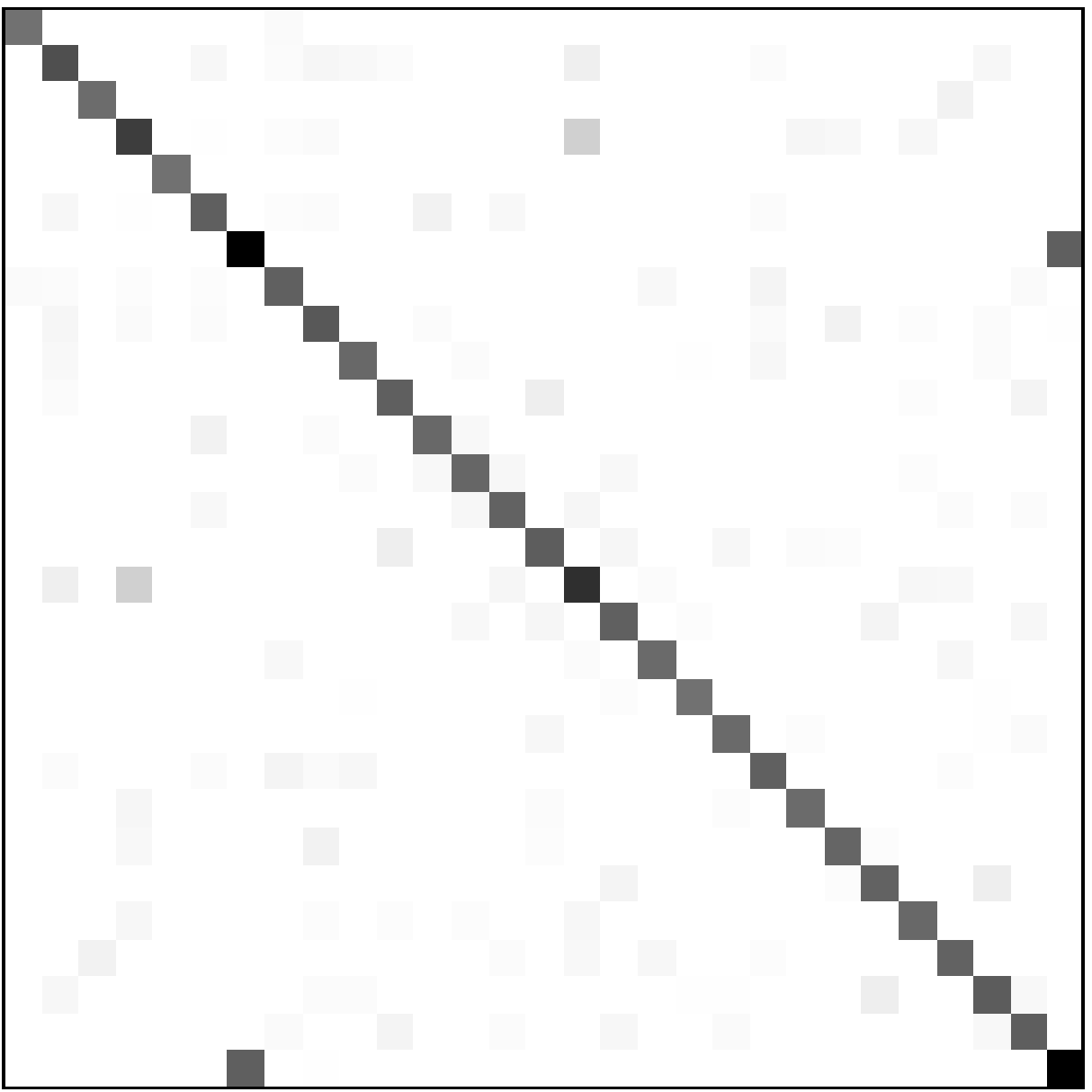}
\caption{Estimates provided by \ours~(left) and \conc~(right); darker means larger in magnitude.}
\label{fig:heatmap}
\end{figure}

\ifwantbiomtabs
\input{tab_short}
\else
\input{tab_short_arxiv}
\fi

%% file: tab_risk_adj_arxiv.tex
\begin{table}[h]
	\centering
	\begin{scriptsize}
	\begin{tabular}{lcccccc}
		\\
		$H$ & \ours & \conc & \samplee & \glasso & \condreg & \ledoit \\[5pt]
		\input{results_for_paper/portopt/rlzd_risk_adj_nest_all_bc_8pct_tc_50basispts}
	\end{tabular}
	\end{scriptsize}
	\caption{Realized risk for various estimators and estimation horizons $H$, after accounting for borrowing and transaction costs, in the portfolio optimization example.  Lower is better; best in \textbf{bold}.  \ours~is best 5/8 times.}
	\label{tab:risk:adj}
\end{table}

%% file: results_for_paper/portopt/rlzd_risk_adj_nest_all_bc_8pct_tc_50basispts.tex
35 & \textbf{14.98} & 16.75 & 33.70 & 16.29 & 17.61 & 15.32\\
40 & \textbf{14.79} & 16.73 & 26.46 & 16.27 & 17.54 & 15.21\\
45 & \textbf{14.98} & 16.75 & 23.13 & 16.28 & 17.43 & 15.21\\
50 & \textbf{14.77} & 16.73 & 20.87 & 16.10 & 17.39 & 15.15\\
75 & \textbf{14.82} & 16.76 & 17.25 & 15.38 & 16.98 & 14.91\\
150 & 14.81 & 16.80 & 15.18 & 14.74 & 16.17 & \textbf{14.45}\\
225 & 14.85 & 16.81 & 14.77 & 14.64 & 15.85 & \textbf{14.29}\\
300 & 14.96 & 16.86 & 14.74 & 14.73 & 15.88 & \textbf{14.29}\\

%% file: tab_sharpe_adj_arxiv.tex
\begin{table}[h]
	\centering
	\begin{scriptsize}
	\begin{tabular}{lcccccc}
		\\
		$H$ & \ours & \conc & \samplee & \glasso & \condreg & \ledoit \\[5pt]
		\input{results_for_paper/portopt/sharpe_adj_nest_all_bc_8pct_tc_50basispts}
	\end{tabular}
	\end{scriptsize}
	\caption{Sharpe ratios for various estimators and estimation horizons $H$, after accounting for borrowing and transaction costs, in the portfolio optimization example.  Higher is better; best in \textbf{bold}.  \ours~is best 4/8 times.}
	\label{tab:sharpe:adj}
\end{table}

%% file: results_for_paper/portopt/sharpe_adj_nest_all_bc_8pct_tc_50basispts.tex
35 & \textbf{0.47} & 0.42 & 0.35 & 0.42 & 0.42 & 0.44\\
40 & 0.46 & 0.42 & \textbf{0.50} & 0.43 & 0.43 & 0.41\\
45 & 0.40 & \textbf{0.41} & 0.30 & 0.40 & 0.41 & 0.36\\
50 & \textbf{0.43} & 0.42 & 0.23 & 0.40 & 0.41 & 0.38\\
75 & \textbf{0.41} & 0.40 & 0.36 & 0.34 & 0.40 & 0.33\\
150 & 0.42 & 0.42 & 0.27 & 0.33 & \textbf{0.43} & 0.36\\
225 & 0.46 & 0.45 & 0.33 & 0.33 & \textbf{0.48} & 0.38\\
300 & \textbf{0.49} & 0.45 & 0.32 & 0.32 & 0.44 & 0.37\\

%% file: tab_short_arxiv.tex
\begin{table}[h]
	\centering
	\begin{scriptsize}
	\begin{tabular}{lcccccc}
		\\
		$H$ & \ours & \conc & \samplee & \glasso & \condreg & \ledoit \\[5pt]
		\input{results_for_paper/portopt/mean_short_side_pct}
	\end{tabular}
	\end{scriptsize}
	\caption{Average size of the short side for various estimators and estimation horizons $H$ in the portfolio optimization example.}
	\label{tab:short}
\end{table}

%% file: results_for_paper/portopt/mean_short_side_pct.tex
35 & 6.91 & 0.06 & 41.13 & 0.63 & 1.77 & 20.50\\
40 & 6.80 & 0.06 & 38.64 & 0.67 & 1.91 & 20.45\\
45 & 6.64 & 0.05 & 36.89 & 0.83 & 2.21 & 20.31\\
50 & 6.60 & 0.04 & 35.46 & 1.36 & 2.43 & 20.33\\
75 & 5.93 & 0.04 & 30.89 & 8.60 & 4.11 & 20.13\\
150 & 5.74 & 0.02 & 25.65 & 23.34 & 7.58 & 19.60\\
225 & 5.59 & 0.01 & 23.68 & 23.35 & 9.34 & 19.26\\
300 & 5.22 & 0.00 & 22.45 & 22.43 & 9.41 & 18.85\\

%% file: supp_conv.tex
\section{Proof of Lemma \ref{thm:conv}}
We prove this result by first establishing, in the following lemma, that the gradient of the smooth term in the objective in the \ours~optimization problem \eqref{eq:ours}, $\nabla g$, is Lipschitz continuous.

\begin{lemma}
\label{lem:lip}
Suppose $\Omega^{(0)}, \ldots, \Omega^{(k)}$ is a sequence of \ours~iterates with nonincreasing objective value.  Let $\Omega$ be any of the iterates here.  Also, let $L = 1/\ell^2 + \| S \|_2 + \lambda_2$, with $\| \cdot \|_2$ denoting the $\ell_2$ operator norm (maximum singular value), and $\ell$ being a constant that uniformly lower bounds $\Omega_{ii}$, $i=1,\ldots,p$.  Then we get that $\nabla^2 g(\Omega) \preceq L I_{p^2 \times p^2}$.
\end{lemma}

\begin{proof}
Let $J_\ourss$ be the objective in the \ours~optimization problem \eqref{eq:ours}.  Then we have that
\[
-\sum_{i=1}^p \log \Omega_{ii} + (\lambda_2/2) \sum_{i=1}^p \Omega^2_{ii} \leq J_\ourss( \Omega^{(0)} ),
\]
since the $\ell_1$ term in the objective in \eqref{eq:ours} is nonnegative, and the trace term can be expressed as a nonnegative quadratic form.  The lefthand side here approaches $\infty$ as either $\Omega_{ii} \rightarrow \infty$ or $\Omega_{ii} \rightarrow 0$, \ie, $\Omega_{ii}$ must be uniformly bounded away from $\infty$ and 0 by some $u$ and $\ell$, respectively, for $i=1,\ldots,p$, owing to the righthand side of the expression.  Thus, we can upper bound the eigenvalues of \eqref{eq:hess} with
\[
1/\ell^2 + \| S \|_2^2 + \lambda_2,
\]
as claimed.
\end{proof}

Obtaining linear convergence is now immediate.  As $g$ is smooth, the conclusion in Lemma \ref{lem:lip} is equivalent to $\nabla^2 g(\Omega) \preceq L I_{p^2 \times p^2} \iff \| \nabla g(\Omega) - \nabla g(\tilde{\Omega}) \|_F \leq L \| \Omega - \tilde{\Omega} \|_F $, where $\tilde{\Omega} \in \symm_{++}^p$, and $L = 1/\ell^2 + \| S \|_2 + \lambda_2$.  Now, since $g$ is also $\lambda_2$-strongly convex, the claim follows by \citet[Proposition 3]{schmidt2011convergence}.  $\; \square$

%% file: supp_sat.tex
\section{Proof of Theorem \ref{thm:sat}}
\begin{proof}
We proceed by first showing that there exists a \conc~estimate that saturates; then we show that the \ours~estimate does not saturate.

A \conc~estimate is defined as a solution to the following (convex) optimization problem:
\optprobstart
\begin{array}{ll}
\minimizewrt{\Omega \in {\bf R}^{p \times p}} & - (1/2) \log \det(\Omega_\D^2) + (n/2) \Tr S \Omega^2 + \lambda_1 \| \Omega_\X \|_1,
\end{array}
\label{eq:conc}
\optprobend
where, as a reminder, $\Omega_\D \in \reals^{p \times p}$ is a matrix of the diagonal entries of $\Omega$, with its off-diagonal entries set to zero; $S \in \reals^{p \times p}$ is the sample covariance matrix, \ie, $S = (1/n) X^T X$, and $X \in \reals^{n \times p}$ is a data matrix; $\Omega_\X \in \reals^{p \times p}$ is a matrix of the off-diagonal entries of $\Omega$, with its diagonal entries set to zero; $\lambda_1$ is a tuning parameter; and $\| \cdot \|_1$ is the elementwise $\ell_1$ norm.

Letting
\begin{equation}
\tilde{J}_\concs(\Omega_\D) = \inf_{\Omega_\X} \; (1/2) \sum_{i=1}^p \left\| \sum_{j=1}^p \Omega_{ij} X_{j} \right\|_2^2 + \lambda_1 \| \Omega_\X \|_1,  \label{eq:qcontilde}
\end{equation}
we see that the optimization problem \eqref{eq:conc} above is equivalent to
\optprobstartnn
\begin{array}{ll}
\minimizewrt{\Omega_\D} & -(1/2) \log \det(\Omega_\D^2) + \tilde{J}_\concs(\Omega_\D).
\end{array}
\optprobendnn

Next, define
\begin{equation*}
b =
\mtxstart{c}
\Omega_{11} X_{1} \\
\Omega_{22} X_{2} \\
\Omega_{33} X_{3} \\
\vdots \\
\Omega_{pp} X_{p}
\mtxend, \quad
\omega = 
\left[
\begin{array}{c}
\Omega_{12} \\
\Omega_{13} \\
\vdots \\
\Omega_{1p} \\
\Omega_{23} \\
\Omega_{24} \\
\vdots \\
\Omega_{2p} \\
\Omega_{34} \\
\Omega_{35} \\
\vdots \\
\Omega_{3p} \\
\vdots \\
\Omega_{p-1,p}
\end{array}
\right],
\end{equation*}
\ie, $b \in \reals^{np}$ and $\omega = \vech \Omega \in \reals^{p(p-1)/2}$.

Then we can express \eqref{eq:qcontilde} as
\begin{equation}
\inf_{\omega} \; (1/2) \| b - A \omega \|_2^2 + \lambda_1 \| \omega \|_1, \label{eq:lasso}
\end{equation}
which is evidently a lasso problem with variable $\omega$.

Then, by \citet[Lemma 14]{tibshirani2013}, for any $b$, $A$, and $\lambda_1 > 0$, there exists a solution $\hat{\omega}(\Omega_\D)$ of \eqref{eq:lasso} (note that we have written here the solution $\hat{\omega}$ as a function of $\Omega_\D$ to emphasize the dependence on $\Omega_\D$) that will have at most $\min \{ np, p(p-1)/2 \}$ nonzero entries for any value of $\Omega_\D$; thus, when $p \gg n$, $\card \hat{\omega}(\Omega_\D) \leq np$, as claimed.  The final claim in the statement of the result follows by invoking \citet[Lemma 3]{tibshirani2013}.

Now, turning to the \ours~optimization problem \eqref{eq:ours}, we have that the trace plus the squared Frobenius norm penalty in the objective in \eqref{eq:ours} can be expressed as
\begin{align}
(n/2) \Tr S \Omega^2 + (\lambda_2 / 2) \sum_{i,j=1}^p \Omega_{ij}^2 & = (1/2) \sum_{i=1}^p \Omega_{i}^T X^T X \Omega_{i} + (\lambda_2 / 2) \sum_{i=1}^p \Omega_{i}^T \Omega_{i} \notag \\
& = (1/2) \sum_{i=1}^p \Omega_{i}^T \left( X^T X + \lambda_2 I \right) \Omega_{i} \notag \\
& = (1/2) \sum_{i=1}^p \left\| \sum_{j=1}^p \Omega_{ij} \mtxstart{c} X_{j} \\ \sqrt{\lambda_2} e_j \mtxend \right\|_2^2, \notag \\
\end{align}
where, as a reminder, $e_i$ is the $i$th standard basis vector in $\reals^p$.

Thus, following a similar argument as above, we can express \eqref{eq:ours} as a lasso problem with variable $\omega \in \reals^{p(p-1)/2}$, $A \in \reals^{ p(n+p) \times p(p-1)/2 }$, and $b \in \reals^{p(n+p)}$; however, in this case, the solution $\hat{\omega}(\Omega_\D)$ can have $p(p-1)/2$ nonzeros, as claimed.
\end{proof}

\section{Proof of Corollary \ref{thm:sat2}}
We prove these results by following a strategy similar to the one we used in the proof of Theorem \ref{thm:sat}.  Note that, at the end of some iteration $i-1$, we can consider the variables $D$ (for \splice) and $\Omega_\D$ (for \spacee) fixed, and then optimize over $B$ (for \splice) and $\Omega_\X$ (for \spacee).  Accordingly, we let (for \splice)
\begin{equation*}
b_{\spls}^{(i-1)} =
\left[
\begin{array}{c}
(1 / \hat D_{11}^{(i-1)}) X_{1} \\
(1 / \hat D_{22}^{(i-1)}) X_{2} \\
(1 / \hat D_{33}^{(i-1)}) X_{3} \\
\vdots \\
(1 / \hat D_{pp}^{(i-1)}) X_{p}
\end{array}
\right], \quad
\omega_{\spls} =
\left[
\begin{array}{c}
(B_{1 \cdot})^T \\
(B_{2 \cdot})^T \\
(B_{3 \cdot})^T \\
\vdots \\
(B_{p \cdot})^T
\end{array}
\right],
\end{equation*}
\begin{equation*}
A_{\spls}^{(i-1)} =
\mtxstart{cccccc}
(1 / \hat D_{11}^{(i-1)}) X_{-1} & 0 & 0 & 0 & \cdots & 0 \\
0 & (1 / \hat D_{22}^{(i-1)}) X_{-2} & 0 & 0 & \cdots & 0 \\
0 & 0 & (1 / \hat D_{33}^{(i-1)}) X_{-3} & 0 & \cdots & 0 \\
\multicolumn{6}{c}{\vdots} \\
0 & 0 & 0 & 0 & \cdots & 
(1 / \hat D_{pp}^{(i-1)}) X_{-p}
\mtxend,
\end{equation*}
\ie, $b_{\spls}^{(i-1)} \in \reals^{np}$, $\omega_{\spls} \in \reals^{p(p-1)}$, and $A_{\spls}^{(i-1)} \in \reals^{np \times p(p-1)}$.  We also let (for \spacee)
\begin{equation*}
b_{\spcs}^{(i-1)} =
\left[
\begin{array}{c}
\sqrt{\hat \Omega_{11}^{(i-1)}} X_{1} \\
\sqrt{\hat \Omega_{22}^{(i-1)}} X_{2} \\
\sqrt{\hat \Omega_{33}^{(i-1)}} X_{3} \\
\vdots \\
\sqrt{\hat \Omega_{pp}^{(i-1)}} X_{p}
\end{array}
\right], \quad
\omega_{\spcs} =
\left[
\begin{array}{c}
\Omega_{12} \\
\Omega_{13} \\
\vdots \\
\Omega_{1p} \\
\Omega_{23} \\
\Omega_{24} \\
\vdots \\
\Omega_{2p} \\
\Omega_{34} \\
\Omega_{35} \\
\vdots \\
\Omega_{3p} \\
\vdots \\
\Omega_{p-1,p}
\end{array}
\right],
\end{equation*}
\ifwantbiomtabs
\else
\begin{tiny}
\fi
\begin{equation*}
A_{\spcs}^{(i-1)} =
\left[
\begin{array}{cccccccccccccccccccccccc}
\tilde{X}{2} & \tilde{X}{3} & \tilde{X}{4} & \cdots & \tilde{X}{p-1} & \tilde{X}{p} &
0 & \multicolumn{16}{c}{\cdots} & 0 \\
\tilde{X}{1} & 0 & \multicolumn{3}{c}{\cdots} & 0 &
\tilde{X}{3} & \tilde{X}{4} & \tilde{X}{5} & \cdots & \tilde{X}{p-1} & \tilde{X}{p} &
0 & \multicolumn{10}{c}{\cdots} & 0 \\
0 & \tilde{X}{1} & 0 & \multicolumn{2}{c}{\cdots} & 0 & 
\tilde{X}{2} & 0 & \multicolumn{3}{c}{\cdots} & 0 &
\tilde{X}{4} & \tilde{X}{5} & \tilde{X}{6} & \cdots & \tilde{X}{p-1} & \tilde{X}{p} &
0 & \multicolumn{4}{c}{\cdots} & 0 \\
\multicolumn{24}{c}{\vdots} \\
0 & \multicolumn{3}{c}{\cdots} & 0 & \tilde{X}{1} & 0 & 
\multicolumn{3}{c}{\cdots} & 0 & \tilde{X}{2} & 0 & \multicolumn{3}{c}{\cdots} & 0 & \tilde{X}{3} & 0 & 
\multicolumn{3}{c}{\cdots} & 0 & \tilde{X}{p-1}
\end{array}
\right],
\end{equation*}
\ifwantbiomtabs
\else
\end{tiny}
\fi
where we write $\tilde{X}_j = \sqrt{\hat \Omega_{jj}^{(i-1)}} X_j$; so, $b_{\spcs}^{(i-1)} \in \reals^{np}$, $\omega_{\spcs} \in \reals^{p(p-1)/2}$, and $A_{\spcs}^{(i-1)} \in \reals^{np \times p(p-1)/2}$.  Applying \citet[Lemma 14]{tibshirani2013} as before, and noting that applying \eqref{eq:splice2} does not affect the sparsity pattern of $\hat{B}^{(i)}$ for \splice, gives the required results.

%% file: supp_cons.tex
\section{Proof of Theorem \ref{thm:cons}}
\begin{proof}
Define $w_i = \hat{\Omega}_{ii}^2, \; i=1,\ldots,p$, where, as a reminder, the $\hat{\Omega}_{ii}$ are estimates of the diagonal entries of $\Omegatrue$ that are assumed in condition (iv) (see the statement of Theorem \ref{thm:cons}), and consider the change of variables for the off-diagonal entries of $\Omega$
\begin{equation*}
\omega_{ij} = - \theta_{ij} ( \hat{\Omega}_{ii} \hat{\Omega}_{jj} )^{1/2}, \quad i,j=1,\ldots,p, \; i \neq j, \label{eq:change}
\end{equation*}
where $\theta \in \reals^{p(p-1)}$ and again $\omega = \vect \Omega$; then we can express the trace term in the objective in the \ours~optimization problem \eqref{eq:ours} as 
\begin{align}
n \Tr S \Omega^2 & = \sum_{i=1}^p ( w_i / \hat{\Omega}_{ii}^2 ) \Omega_i^T X^T X \Omega_i \notag \\
& = \sum_{i=1}^p ( w_i / \hat{\Omega}_{ii}^2 ) \left\| \sum_{j=1}^p \omega_{ij} X_j \right\|_2^2 \notag \\
& = \sum_{i=1}^p w_i \left\| (1/\hat{\Omega}_{ii}) \left( \hat{\Omega}_{ii} X_i + \sum_{j \neq i}^p \omega_{ij} X_j \right) \right\|_2^2 \notag \\
& = \sum_{i=1}^p w_i \left\| X_i + \sum_{j \neq i}^p ( \omega_{ij} / \hat{\Omega}_{ii} ) X_j \right\|_2^2 \notag \\
& = \sum_{i=1}^p w_i \left\| X_i - \sum_{j \neq i}^p \theta_{ij} \left( \hat{\Omega}_{jj} / \hat{\Omega}_{ii} \right)^{1/2} X_j \right\|_2^2. \label{eq:space2}
\end{align}

Equation \ref{eq:space2} is equal to the objective of the \spacee~optimization problem (\cf~\citet[Equation 10]{peng2009partial} and/or the trace term in \citet[Equation 12]{khare2014convex}), up to constants and for fixed diagonal entries; thus, the $\log \det$ term (which is only a function of diagonal entries) plus the trace term in the objective in \eqref{eq:ours} are also equivalent to the corresponding terms in the \spacee's objective.  This implies that properties A1--A4 and B0--B3 in the supplement for \cite{peng2009partial} also apply to the $\log \det$ plus trace terms in the objective in \eqref{eq:ours}.

Now, let $L(\theta)$ denote the $\log \det$ plus trace terms in the objective in \eqref{eq:ours} (with variable off-diagonal entries $\theta \in {\bf R}^{p(p-1)}$ and fixed diagonal entries $\hat{\omega}_\D$), and let $B_{c_1}(\theta_\X^0, c_1 q_n^{1/2} \lambda_{1,n})$ be a ball of radius $c_1 q_n^{1/2} \lambda_{1,n}$, for a constant $c_1 > 0$, with center $\theta_\X^0$, \ie, $B_{c_1} = \{ \theta : \| \theta - \theta^0_\X \|_2 \leq c_1 q_n^{1/2} \lambda_{1,n} \}$, where $\theta_\X^0$ is the application of the same (strictly monotone) transformation in \eqref{eq:change} to the underlying off-diagonal entries $\omega_\X^0$.

First, we show that the unique, global solution (owing to the strong convexity of \eqref{eq:restr}) of the following ``restricted'' optimization problem lies in $B_{c_1}$ with probability tending to one as $n \to \infty$:
\begin{equation}
\begin{array}{ll}
\minimizewrt{ \theta : \theta_{\mathcal{A}_n^c} = 0 } & L(\theta) + \lambda_{1,n} \sum_{i \neq j}^p \left| ( \hat{\Omega}_{ii} \hat{\Omega}_{jj} )^{1/2} \theta_{ij} \right| + (\lambda_{2,n} / 2) \sum_{i \neq j}^p \hat{\Omega}_{ii} \hat{\Omega}_{jj} \theta_{ij}^2.
\end{array}
\label{eq:restr}
\end{equation}

Let $\alpha_n = q_n^{1/2} \lambda_{1,n}$, and let $u \in \reals^{p(p-1)}$ with $u_{\mathcal{A}_n^c} = 0$ and $\| u \|_2 = c$, for a constant $c > 0$.  Fix $\theta \in B_{c_1}$ to be equal to $\theta_\X^0 + \alpha_n u$.  Then we have that
\begin{align}
& \lambda_{1,n} \left( \sum_{i \neq j}^p \left| ( \hat{\Omega}_{ii} \hat{\Omega}_{jj} )^{1/2} \theta_{\X, ij}^0 \right| - \sum_{i \neq j}^p \left| ( \hat{\Omega}_{ii} \hat{\Omega}_{jj} )^{1/2} \theta_{ij} \right| \right) \notag \\
& \quad \leq \lambda_{1,n} \sum_{i \neq j}^p \left| ( \hat{\Omega}_{ii} \hat{\Omega}_{jj} )^{1/2} ( \theta_{\X, ij}^0 - \theta_{ij} ) \right| \notag \\
& \quad = \lambda_{1,n} \alpha_n \sum_{i \neq j}^p \left| ( \hat{\Omega}_{ii} \hat{\Omega}_{jj} )^{1/2} u_{ij} \right| \notag \\
& \quad = O( \lambda_{1,n} \alpha_n q_n^{1/2} \| u \|_2 ) \notag \\
& \quad = O(\alpha_n^2), \label{eq:onenormub}
\end{align}
with probability at least $1 - O(n^{-\beta})$, as the diagonal estimates $\hat{\Omega}_{ii}$ are uniformly bounded with high probability; the second line here follows by the triangle inequality, the third by the choice of $\theta$, the fourth by the Cauchy-Schwarz inequality and the definition of $u$, and the fifth by the definition $\alpha_n = q_n^{1/2} \lambda_{1,n}$.

We also have that
\begin{align}
& (\lambda_{2,n} / 2) \left( \sum_{i \neq j}^p \hat{\Omega}_{ii} \hat{\Omega}_{jj} (\theta_{\X, ij}^0)^2 - \sum_{i \neq j}^p \hat{\Omega}_{ii} \hat{\Omega}_{jj} \theta_{ij}^2 \right) \label{eq:twonormub:5} \\
& \quad = (\lambda_{2,n} / 2) \left( \sum_{i \neq j}^p \hat{\Omega}_{ii} \hat{\Omega}_{jj} (\theta_{\X, ij}^0)^2 - \sum_{i \neq j}^p \hat{\Omega}_{ii} \hat{\Omega}_{jj} (\theta_{\X, ij}^0 + \alpha_n u_{ij})^2 \right) \notag \\
& \quad = - \lambda_{2,n} \alpha_n \sum_{i \neq j}^p \hat{\Omega}_{ii} \hat{\Omega}_{jj} \theta_{\X, ij}^0 u_{ij} - (\lambda_{2,n} / 2) \alpha_n^2 \sum_{i \neq j}^p \hat{\Omega}_{ii} \hat{\Omega}_{jj} u_{ij}^2. \label{eq:twonormub:2}
\end{align}

We get for the first term in \eqref{eq:twonormub:2} that
\begin{align}
- \lambda_{2,n} \alpha_n \sum_{i \neq j}^p \hat{\Omega}_{ii} \hat{\Omega}_{jj} \theta_{\X, ij}^0 u_{ij} & \leq O( \lambda_{2,n} \alpha_n q_n^{1/2} ) \| u \|_2 \notag \\
& = o(\alpha_n^2) \| u \|_2, \label{eq:twonormub:3}
\end{align}
with probability at least $1 - O(n^{-\beta})$; the first line here follows by the Cauchy-Schwarz inequality, and the second by the assumption that $\lambda_{2,n} = o(\lambda_{1,n})$.

Similarly, we get for the second term in \eqref{eq:twonormub:2}
\begin{align}
- (\lambda_{2,n} / 2) \alpha_n^2 \sum_{i \neq j}^p \hat{\Omega}_{ii} \hat{\Omega}_{jj} u_{ij}^2 & \leq o(\alpha_n^2) \| u \|_2^2, \label{eq:twonormub:4}
\end{align}
with probability at least $1 - O(n^{-\beta})$.

Putting \eqref{eq:twonormub:3} and \eqref{eq:twonormub:4} together, we get for \eqref{eq:twonormub:5} that
\begin{align}
& (\lambda_{2,n} / 2) \left( \sum_{i \neq j}^p \hat{\Omega}_{ii} \hat{\Omega}_{jj} (\theta_{\X, ij}^0)^2 - \sum_{i \neq j}^p \hat{\Omega}_{ii} \hat{\Omega}_{jj} \theta_{ij}^2 \right) \leq o(\alpha_n^2) \left( \| u \|_2 + \| u \|_2^2 \right) \label{eq:twonormub}
\end{align}
with probability at least $1 - O(n^{-\beta})$.

Next, let $J_\ourss(\theta)$ equal the objective in \eqref{eq:ours} (with fixed diagonal entries $\hat{\omega}_\D$); combining \eqref{eq:onenormub} and \eqref{eq:twonormub}, we get
\begin{align*}
J_\ourss(\theta) - J_\ourss(\theta^0_\X) & \geq L(\theta) - L(\theta_\X^0) \\
& \quad - \lambda_{1,n} \left( \sum_{i \neq j}^p \left| ( \hat{\Omega}_{ii} \hat{\Omega}_{jj} )^{1/2} \theta_{\X, ij}^0 \right| - \sum_{i \neq j}^p \left| ( \hat{\Omega}_{ii} \hat{\Omega}_{jj} )^{1/2} \theta_{ij} \right| \right) \\
& \quad - (\lambda_{2,n} / 2) \left( \sum_{i \neq j}^p \hat{\Omega}_{ii} \hat{\Omega}_{jj} (\theta_{\X, ij}^0)^2 - \sum_{i \neq j}^p \hat{\Omega}_{ii} \hat{\Omega}_{jj} \theta_{ij}^2 \right) \\
& \geq L(\theta) - L(\theta_\X^0) - O(\alpha_n^2) - o(\alpha_n^2) \\
& = L(\theta) - L(\theta_\X^0) - O(\alpha_n^2).
\end{align*}

By the same arguments in the proof of Lemma S-3 in the supplement for \citet{peng2009partial}, it follows that the (unique, global) solution to the restricted problem \eqref{eq:restr} lies in $B_{c_1}$, with probability at least $1 - O(n^{-\beta})$; this also implies (by a simple contradiction argument) that the event $\sign \hat{\theta}_{\mathcal{A}_n} = \sign \theta^0_{\mathcal{A}_n}$ occurs with high probability.

By construction, the solution $\hat{\theta}$ to the restricted optimization problem \eqref{eq:restr} satisfies the support ``block'' of the optimality conditions for the unrestricted optimization problem \eqref{eq:ours}.  Next, we show that $\hat{\theta}$ satisfies the non-support (the complement of the support) block of the optimality conditions for the unrestricted optimization problem \eqref{eq:ours}.

The optimality conditions for the unrestricted optimization problem \eqref{eq:ours} are
\begin{equation}
\begin{array}{lll}
L_{ij}^{'}(\theta) + \lambda_{2,n} \hat{\Omega}_{ii} \hat{\Omega}_{jj} \theta_{ij} & = - \lambda_{1,n} ( \hat{\Omega}_{ii} \hat{\Omega}_{jj} )^{1/2} \sign \theta_{ij} & \quad \textrm{if } \theta_{ij} \neq 0 \\
| L_{ij}^{'}(\theta) + \lambda_{2,n} \hat{\Omega}_{ii} \hat{\Omega}_{jj} \theta_{ij} | & \leq \lambda_{1,n} ( \hat{\Omega}_{ii} \hat{\Omega}_{jj} )^{1/2} & \quad \textrm{if } \theta_{ij} = 0,
\end{array}
\end{equation}
where $L^{'}_{ij}(\theta) = \partial L(\theta) / \partial \theta_{ij}$; this establishes the analog of Lemma S-1 in the supplement for \citet{peng2009partial}, and also implies that Lemma S-2 there applies to the unrestricted optimization problem \eqref{eq:ours} here.  We wish to show that (with high probability)
\[
\max_{(i,j) \in \mathcal{A}_n^c} | L_{ij}^{'}(\hat{\theta}) + \lambda_{2,n} \hat{\Omega}_{ii} \hat{\Omega}_{jj} \hat{\theta}_{ij} | < \lambda_{1,n} ( \hat{\Omega}_{ii} \hat{\Omega}_{jj} )^{1/2}.
\]

We begin by taking an exact (since $L^{'}_{\mathcal{A}_n}$ is affine) first-order Taylor expansion of $L^{'}_{\mathcal{A}_n}(\hat{\theta})$ around $\theta^0$, \ie,
\begin{align}
L^{'}_{\mathcal{A}_n}(\hat{\theta}) & = L^{'}_{\mathcal{A}_n}(\theta^0) + L^{''}_{\mathcal{A}_n \mathcal{A}_n} \underbrace{( \hat{\theta} - \theta^0 )}_{v}  \notag \\
& = L^{'}_{\mathcal{A}_n}(\theta^0) + \underbrace{( L^{''}_{\mathcal{A}_n \mathcal{A}_n}(\theta^0) - \bar{L}^{''}_{\mathcal{A}_n \mathcal{A}_n}(\theta^0) )}_{ \Delta_{\mathcal{A}_n \mathcal{A}_n} } v + \bar{L}^{''}_{\mathcal{A}_n \mathcal{A}_n}(\theta^0) v. \label{eq:cons:grad}
\end{align}

However, we also have that, with probability at least $1 - O(n^{-\beta})$,
\begin{equation}
L^{'}_{\mathcal{A}_n}(\hat{\theta}) = - \lambda_{1,n} ( \hat{\Omega}_{ii} \hat{\Omega}_{jj} )^{1/2} \sign \theta_{\mathcal{A}_n}^0 \label{eq:cons:grad2}.
\end{equation}

Equating \eqref{eq:cons:grad} and \eqref{eq:cons:grad2} and rearranging, we get
\begin{equation}
v = - \left( \bar{L}^{''}_{\mathcal{A}_n \mathcal{A}_n}(\theta^0) \right)^{-1} \left( \lambda_{1,n} ( \hat{\Omega}_{ii} \hat{\Omega}_{jj} )^{1/2} \sign \theta_{\mathcal{A}_n}^0 + L^{'}_{\mathcal{A}_n}(\theta^0) + \Delta_{\mathcal{A}_n \mathcal{A}_n} v \right). \label{eq:v}
\end{equation}

Repeating a similar analysis for any $(i,j) \in \mathcal{A}_n^c$, we get
\begin{equation}
L^{'}_{ij}(\hat{\theta}) = L^{'}_{ij}(\theta^0) + \Delta_{ij, \mathcal{A}_n}(\theta^0) v + \bar{L}^{''}_{ij, \mathcal{A}_n}(\theta^0) v. \label{eq:cons:grad3}
\end{equation}

Now, plugging \eqref{eq:v} into the third term on the righthand side of \eqref{eq:cons:grad3}, we get
\begin{align*}
L^{'}_{ij}(\hat{\theta}) & = L^{'}_{ij}(\theta^0) + \Delta_{ij, \mathcal{A}_n}(\theta^0) v \\
& \quad - \lambda_{1,n} ( \hat{\Omega}_{ii} \hat{\Omega}_{jj} )^{1/2} \bar{L}^{''}_{ij, \mathcal{A}_n}(\theta^0) \left( \bar{L}^{''}_{\mathcal{A}_n \mathcal{A}_n}(\theta^0) \right)^{-1} \sign \theta_{\mathcal{A}_n}^0 \\
& \quad - \bar{L}^{''}_{ij, \mathcal{A}_n}(\theta^0) \left( \bar{L}^{''}_{\mathcal{A}_n \mathcal{A}_n}(\theta^0) \right)^{-1} L^{'}_{\mathcal{A}_n}(\theta^0) \\
& \quad - \bar{L}^{''}_{ij, \mathcal{A}_n}(\theta^0) \left( \bar{L}^{''}_{\mathcal{A}_n \mathcal{A}_n}(\theta^0) \right)^{-1} \Delta_{\mathcal{A}_n \mathcal{A}_n} v.
\end{align*}

Applying the triangle inequality and rearranging yields
\begin{align*}
| L^{'}_{ij}(\hat{\theta}) | & \leq \left| \lambda_{1,n} ( \hat{\Omega}_{ii} \hat{\Omega}_{jj} )^{1/2} \bar{L}^{''}_{ij, \mathcal{A}_n}(\theta^0) \left( \bar{L}^{''}_{\mathcal{A}_n \mathcal{A}_n}(\theta^0) \right)^{-1} \sign \theta_{\mathcal{A}_n}^0 \right| \\
& \quad + \left| \left( \Delta_{ij, \mathcal{A}_n}(\theta^0) - \bar{L}^{''}_{ij, \mathcal{A}_n}(\theta^0) \left( \bar{L}^{''}_{\mathcal{A}_n \mathcal{A}_n}(\theta^0) \right)^{-1} \Delta_{\mathcal{A}_n \mathcal{A}_n} \right) v \right| \\
& \quad + \left| \bar{L}^{''}_{ij, \mathcal{A}_n}(\theta^0) \left( \bar{L}^{''}_{\mathcal{A}_n \mathcal{A}_n}(\theta^0) \right)^{-1} L^{'}_{\mathcal{A}_n}(\theta^0) \right| \\
& \quad + | L^{'}_{ij}(\theta^0) |.
\end{align*}

The first term here is (strictly) less than $\lambda_{1,n} ( \hat{\Omega}_{ii} \hat{\Omega}_{jj} )^{1/2} / 2$ by condition (iii), and the remaining terms are $o(\lambda_{1,n})$, with probability at least $1 - O(n^{-\beta})$, by the same arguments in the proof of \citet[Theorem 2]{peng2009partial}.

Now, let $R_{ij}^{'}(\theta) = \lambda_{2,n} \hat{\Omega}_{ii} \hat{\Omega}_{jj} \theta_{ij}$; repeating a similar analysis as above, we get
\begin{align*}
R_{ij}^{'}(\hat{\theta}) & = R_{ij}^{'}(\theta^0) + \left( R_{ij, \mathcal{A}_n}^{''}(\theta^0) - \bar{R}_{ij, \mathcal{A}_n}^{''}(\theta^0) \right) v + \bar{R}^{''}_{ij, \mathcal{A}_n}(\theta^0) v \\
& = R_{ij}^{'}(\theta^0) + \bar{R}^{''}_{ij, \mathcal{A}_n}(\theta^0) v \\
& = \lambda_{2,n} \hat{\Omega}_{ii} \hat{\Omega}_{jj} \theta^0_{ij} + \lambda_{2,n} \hat{\Omega}_{ii} \hat{\Omega}_{jj} v_{ij} \\
& \leq o(\lambda_{1,n}) + \lambda_{2,n} \hat{\Omega}_{ii} \hat{\Omega}_{jj} c_1 q_n^{1/2} \lambda_{1,n} \\
& = o(\lambda_{1,n}),
\end{align*}
where the penultimate line follows since $\| v \|_2 = \| \hat{\theta} - \theta^0 \|_2 \leq c_1 q_n^{1/2} \lambda_{1,n} \implies v_{ij} \leq c_1 q_n^{1/2} \lambda_{1,n}$, and the last line since $q_n^{1/2} \lambda_{1,n} \to 0$ by condition (v).

Putting these findings together, we get, with probability at least $1 - O(n^{-\beta})$,
\[
\max_{(i,j) \in \mathcal{A}_n^c} | L_{ij}^{'}(\hat{\theta}) + R^{'}_{ij}(\hat{\theta}) | < \lambda_{1,n} ( \hat{\Omega}_{ii} \hat{\Omega}_{jj} )^{1/2} / 2 + o(\lambda_{1,n}),
\]as required.

Thus, since the (unique, global) solution to the restricted optimization problem \eqref{eq:restr} satisfies the optimality conditions for the unrestricted optimization problem \eqref{eq:ours} (which also admits a unique, global solution), and since the restricted solution lies in $B_{c_1}$, we obtain the required results.
\end{proof}

%% file: supp_diags.tex
\section{Proof of Theorem \ref{thm:diags}}
We start by considering the estimation of the $p$th diagonal entry for ease of exposition. As discussed later, the argument below 
(all the way to Equation \eqref{eq39}) can be repeated verbatim for estimation of the $i$th diagonal entry with obvious 
notational changes.

Note that, since $d_n = O(q_n)$, conditions (i), (ii), (v), and (vi) imply that $d_n^{1/2} \lambda_{1,n} \to 0$, $d_n (\log n / n)^{1/2} \to 0$, and $(1 / \lambda_{1,n}) ( (d_n / n) \log n )^{1/2} \to 0$.

Let $(\eta^T, 1) = \Omega_{p \cdot} / \Omega_{pp}$, \ie, $\eta$ is the 
$p$th (off-diagonal) row of $\Omega$ divided by the $p$th diagonal entry. Let $S$ again denote the sample 
covariance matrix.  Consider the function 
$$
J_p(\eta) =  (\eta^T, 1) S (\eta^T, 1)^T + \lambda_{1,n} \sum_{i=1}^{p-1} |\eta_i|, 
$$
where again $\lambda_{1,n}$ is the tuning parameter. This a convex function, and any global minimizer of this function will be sparse 
in $\eta$. This will immediately lead to an estimate of the sparsity in the $p$th row of $\Omega$. The function 
$J_p$ is the same objective function used by \citet{meinshausen2006high} in their 
neighborhood selection procedure (up to a simple transformation of the parameter $\eta$). Note that \citet{meinshausen2006high} provide 
a consistency proof for the sparsity pattern obtained by minimizing $J_p$ under a set of regularity assumptions (for example, 
Gaussianity).\footnote{Note that, by combining the sparsity patterns for all the rows of $\Omega$ using the neighborhood 
selection procedure, one can obtain an estimate for the sparsity pattern in $\Omegatrue$. However, a drawback is that the resulting 
pattern is not necessarily symmetric. On the other hand, our goal in this section is to show consistency of a procedure, which 
uses the sparsity pattern for neighborhood selection solely for estimating the diagonal entries of $\Omegatrue$.}
We provide a proof of sparsity selection consistency for $J_p$ below under a set of related but different assumptions from those 
in \citet{meinshausen2006high} (for example, under a general sub-Gaussian tail setting). 

Let $\eta^0$ denote the true value of the parameter $\eta$. Also, for ease of exposition, we use 
$\eta_p = \eta^0_p = 1$ below, but the vector $\eta$ will always refer to the $(p-1)$-dimensional parameter 
defined above. We now obtain the required result through a sequence of lemmas.

\begin{lemma} \label{lem:consistency:1}
For any $\gamma > 0$, there exists a constant $C_{\gamma} > 0$ such that, with probability at least $1 - 
O(n^{-\gamma})$,
$$
\max_{1 \leq i,j, \leq p} |S_{ij} - \Sigma^0_{ij}| \leq C_{\gamma} \sqrt{\frac{\log n}{n}},
$$
for large enough $n$. 
\end{lemma}

\begin{proof}
Fix $1 \leq i,j \leq p$. Let $\mu_+ = \Expect_{\Sigma^0_n} \left[ (X_{1i} + X_{1j})^2 \right]$ and 
$\mu_- =  \Expect_{\Sigma^0_n} \left[ (X_{1i} - X_{1j})^2 \right]$. It follows that 
\begin{align}
& \P(|S_{ij} - \Sigma^0_{ij}| > t) \nonumber\\
& \quad = \P \left( \left| \frac{1}{n} \sum_{\ell=1}^n (X_{\ell i} + X_{\ell j})^2 - (X_{\ell i} - X_{\ell j})^2 - (\mu_+ - \mu_-) 
\right| > 4t \right) \nonumber\\
& \quad \leq \P \left( \left| \frac{1}{n} \sum_{\ell=1}^n (X_{\ell i} + X_{\ell j})^2 - \mu_+ \right| > 2t \right) + 
\P \left( \left| \frac{1}{n} \sum_{\ell=1}^n (X_{\ell i} - X_{\ell j})^2 - \mu_- \right| > 2t \right). \label{eqbound1} 
\end{align}

Note that $X_{\ell i} + X_{\ell j}$ are sub-Gaussian random variables (by condition (i)), and their variances 
are uniformly bounded in $i$, $j$, and $n$ (by condition (ii)). For any $c_3 > 0$, it follows, by 
(\ref{eqbound1}) and \citet[Theorem 1.1]{rudelson2013hanson}, that there exist constants $K_1$ and 
$K_2$ independent of $i$, $j$, and $n$ such that 
$$
\P \left( |S_{ij} - \Sigma^0_{ij}| > C \sqrt{\frac{\log n}{n}} \right) \leq K_1 e^{-K_2 n \left( c_3 
\sqrt{\frac{\log n}{n}} \right)^2} = K_1 e^{-K_2 C^2 \log n}, 
$$
for large enough $n$. Using the union bound and the fact that $p = O(n^\kappa)$, for some $\kappa > 0$, 
gives us the required result.
\end{proof}

Next, let 
$$
\tilde{L} (\eta) =  (\eta^T, 1) S (\eta^T, 1)^T,
$$
and let 
\begin{equation} \label{eqbound2}
d_i (\eta) = 2 \sum_{j=1}^p \eta_j S_{ij},
\end{equation}
for $1 \leq i \leq p-1$, denote the elements of the gradient of $\tilde{L}$.  Then we obtain the following results.

\begin{lemma}[Optimality conditions] \label{lem:consistency:2}
$\eta$ minimizes $J_p$ if and only if
\begin{equation}
\begin{array}{lll}
d_i (\eta) & = - \lambda_{1,n} \sign \eta_i & \quad \textrm{if } \eta_i \neq 0, \; 1 \leq i \leq p-1 \\
|d_i (\eta)| & \leq \lambda_{1,n} & \quad \textrm{if } \eta_i = 0, 1 \leq i \leq p-1. \label{eq16}
\end{array}
\end{equation}

Also, if $|d_i (\hat\eta)| < \lambda_{1,n}$, for any minimizer $\hat\eta$, then by 
the continuity of $d_i$ and the convexity of $J_p$, it follows that $\tilde{\eta}_i = 0$, for every 
minimizer $\tilde\eta$ of $J_p$. 
\end{lemma}

\begin{lemma} \label{lem:consistency:2.1}
For every $1 \leq i \leq p-1$,
$$
{\Expect}_{\Sigma^0_n} \left[ d_i (\eta^0) \right] = 0. 
$$
\end{lemma}

\begin{proof}
Let $\Sigma^0_{r}$ denote the submatrix of $\Sigma^0$ formed by using the first $r$ rows and 
columns. It follows, by the definition of $\eta^0$, that, for every $1 \leq i < p$,
$$
{\Expect}_{\Sigma^0} \left[ d_i (\eta^0) \right] = 2 \sum_{j=1}^p \eta^0_j \Sigma^0_{ij} = 
\frac{2}{\Omegatrue_{pp}} \sum_{j=1}^p (\Sigma^0)^{-1}_{pj} \Sigma^0_{ij} = 
0. 
$$
\end{proof}

\begin{lemma} \label{lem:consistency:3}
For any $\gamma > 0$, there exists a constant $C_{1, \gamma} > 0$ such that, with probability at least $1 - 
O(n^{-\gamma})$, 
$$
\max_{1 \leq i \leq p} |d_i (\eta^0)|  \leq C_{1, \gamma} \sqrt{\frac{\log n}{n}}. 
$$
\end{lemma}

\begin{proof}
It follows, by Lemma \ref{lem:consistency:2}, that 
$$
d_i (\eta^0) = \frac{2}{n} \sum_{\ell=1}^n X_{\ell i} \left( \sum_{j=1}^p \eta^0_j X_{\ell j} \right) 
$$
is the difference between the sample covariance and population covariance of $X_i$ and $\sum_{j=1}^p 
\eta^0_j X_j$. It follows, by condition (ii) and the definition of $\eta^0$, that the variance of 
$\sum_{j=1}^p \eta^0_j X_j$, given by $\left( (\eta^0)^T, 1 \right) \Sigma^0 \left( (\eta^0)^T, 1 \right)^T$, 
is uniformly bounded over $n$. The proof now follows along the same lines as the proof of 
Lemma \ref{lem:consistency:1}.
\end{proof}

Note that $\mathcal{A}_n^p$ is the set of indices 
corresponding to the nonzero entries of $\eta^0_n$. Also note that $|\mathcal{A}_n^p| \leq d_n$.  Next, we
establish properties for the following ``restricted'' minimization problem:
\begin{equation}
\label{eq20}
\begin{array}{ll}
\minimizewrt{ \eta : \eta_j = 0, \; j \notin \mathcal{A}_n^p } & J_p (\eta).
\end{array}
\end{equation}

\begin{lemma} \label{lem:consistency:6.1}
There exists $C > 0$ such that, for any $\gamma > 0$, a global minimum of the restricted minimization problem (\ref{eq20}) exists within the ball $\{\eta: \|\eta - \eta^0\|_2 < 
C \sqrt{d_n} \lambda_{1,n}\}$, with probability at least $1 - O(n^{-\gamma})$ for sufficiently large $n$. 
\end{lemma}

\begin{proof}
Let $\tilde{\alpha}_n = \sqrt{d_n} \lambda_{1,n}$. Then, for any constant $C > 0$ and any 
$u \in \reals^{p-1}$ satisfying $u_j = 0$ for every $j \notin \mathcal{A}_n^p$ and $\|u\|_2 = C$, we 
get by the triangle inequality that 
\begin{equation} \label{eq21}
\sum_{j=1}^{p-1} |\eta^0_j| - \sum_{j=1}^{p-1} |\eta^0_j + \tilde{\alpha}_n u_j| \leq \tilde{\alpha}_n \sum_{j=1}^{p-1} 
|u_j| \leq C \tilde{\alpha}_n \sqrt{d_n}. 
\end{equation}

Again, let
$$
\tilde{L} (\eta) =  (\eta^T, 1)^T S (\eta^T, 1)^T. 
$$

By (\ref{eq21}) and a second-order Taylor series expansion around $\eta^0$, we get 
\begin{align}
& J_p (\eta^0 + \tilde{\alpha}_n u) - J_p (\eta^0) 
\nonumber\\
& \quad = \tilde{L} (\eta^0 + \tilde{\alpha}_n u) - \tilde{L} (\eta^0)  - \lambda_{1,n} \left( 
\sum_{j=1}^{p-1} |\eta^0_j| - \sum_{j=1}^{p-1} |\eta^0_j + \tilde{\alpha}_n u_j| \right) \nonumber\\
& \quad \geq \tilde{\alpha}_n \sum_{j \in \mathcal{A}_n^p} u_j d_j (\eta^0) + \tilde{\alpha}_n^2 
\sum_{j \in \mathcal{A}_n^p} \sum_{k \in \mathcal{A}_n^p} u_j u_k S_{jk} - C \tilde{\alpha}_n \sqrt{d_n} \lambda_{1,n} \nonumber\\
& \quad \geq \tilde{\alpha}_n \sum_{j \in \mathcal{A}_n^p} u_j d_j (\eta^0) + \tilde{\alpha}_n^2 
\sum_{j \in \mathcal{A}_n^p} \sum_{k \in \mathcal{A}_n^p} u_j u_k (S_{jk} - \Sigma^0_{jk}) + \tilde{\alpha}_n^2 \sum_{j \in 
\mathcal{A}_n^p} \sum_{k \in \mathcal{A}_n^p} u_j u_k \Sigma^0_{ jk} - C \tilde{\alpha}_n^2. \label{eq22}
\end{align}

Note that $\lambda_{1,n} \sqrt{\frac{n}{\log n}} \rightarrow \infty$ and $d_n \sqrt{\frac{\log n}{n}} \rightarrow 0$ as 
$n \rightarrow \infty$, since $(1 / \lambda_{1,n}) ( (d_n / n) \log n )^{1/2} \to 0$ and $d_n^{1/2} \lambda_{1,n} \to 0$. It follows, by the Cauchy-Schwarz inequality, Lemma \ref{lem:consistency:1}, 
and Lemma \ref{lem:consistency:3}, that for any $\gamma > 0$ there exist constants $C_\gamma$ and 
$C_{1, \gamma} > 0$ such that, with probability at least $1 - O(n^{-\gamma})$, 
\begin{equation} \label{eq23}
\tilde{\alpha}_n \sum_{j \in \mathcal{A}_n^p} u_j d_j (\eta^0) \leq C C_{1, \gamma} 
\sqrt{\frac{d_n \log n}{n}} \tilde{\alpha}_n = o(\tilde{\alpha}_n^2) 
\end{equation}
and 
\begin{equation} \label{eq24}
\frac{\tilde{\alpha}_n^2}{2} \left| \sum_{j \in \mathcal{A}_n^p} \sum_{k \in \mathcal{A}_n^p} u_j u_k (S_{jk} - \Sigma^0_{jk}) \right| 
\leq C_\gamma C^2 d_n \sqrt{\frac{\log n}{n}} = o(\tilde{\alpha}_n^2). 
\end{equation}

Also, by condition (ii), it follows that 
\begin{equation} \label{eq25}
\sum_{j \in \mathcal{A}_n^p} \sum_{k \in \mathcal{A}_n^p} u_j u_k \Sigma^0_{ jk} \geq \frac{C^2 
\tilde{\alpha}_n^2}{2 \lambda_{\max}(\Omegatrue)}. 
\end{equation}

Combining (\ref{eq22}), (\ref{eq23}), (\ref{eq24}), and (\ref{eq25}), we get that 
$$
J_p (\eta^0 + \tilde{\alpha}_n u) - J_p (\eta^0)  > 
\frac{C^2 \tilde{\alpha}_n^2}{2 \lambda_{\max}(\Omegatrue)} - 2C \tilde{\alpha}_n^2, 
$$
with probability at least $1 - O(n^{-\gamma})$, for large enough $n$.

Choosing $C = 4 \lambda_{\max}(\Omegatrue) + 1$, 
we obtain that 
$$
\inf_{ u: u_{ \left( {\mathcal{A}_n^p} \right)^c } = 0, \; \|u\|_2 = C } J_p (\eta^0 + 
\tilde{\alpha}_n u) > J_p (\eta^0),  
$$
with probability at least $1 - O(n^{-\gamma})$, for large enough $n$. Hence, for every $\eta > 0$, a local 
minimum (in fact a global minimum due to convexity) of the restricted minimization problem (\ref{eq20}) exists 
within the ball $\{\eta: \|\eta - \eta^0\|_2 < C 
\sqrt{d_n} \lambda_{1,n}\}$, with probability at least $1 - O(n^{-\eta})$, for sufficiently large $n$. 
\end{proof}

\begin{lemma} \label{lem:consistency:6}
There exists a constant $C_1 > 0$ such that, for any $\gamma > 0$, the following holds with probability at least $1 - O(n^{-\gamma})$.

For any $\eta$ in the set 
$$
S = \{\eta: \|\eta - \eta^0\|_2 \geq C_1 \sqrt{d_n} 
\lambda_{1,n}, \; \eta_j = 0 \; \forall j \notin \mathcal{A}_n^p\}, 
$$
we have $\left\| d_{\mathcal{A}_n^p} (\eta) \right\|_2 > \sqrt{d_n} \lambda_{1,n}$, 
where $d_{\mathcal{A}_n^p} (\eta) = \left( d_j (\eta) \right)_{j \in 
\mathcal{A}_n^p}$. 
\end{lemma}

\begin{proof}
Recall that $\tilde{\alpha}_n = \sqrt{d_n} \lambda_{1,n}$. Choose $\eta \in S$ 
arbitrarily. Let $u = \eta - \eta^0/\tilde{\alpha}_n$. It follows that $u_j = 0$, 
for every $j \notin \mathcal{A}_n^p$ and $\|u\| \geq C_1$. By a first-order Taylor series 
expansion of $d_{\mathcal{A}_n^p}$, it follows that 
\begin{eqnarray}
d_{\mathcal{A}_n^p} (\eta) 
&=& d_{\mathcal{A}_n^p} (\eta^0) + 2 \tilde{\alpha}_n S_{\mathcal{A}_n^p \mathcal{A}_n^p} 
u_{\mathcal{A}_n^p} \nonumber\\
&=& d_{\mathcal{A}_n^p} (\eta^0) + 2 \tilde{\alpha}_n \Sigma^0_{ \mathcal{A}_n^p 
\mathcal{A}_n^p} u_{\mathcal{A}_n^p} + 2 \tilde{\alpha}_n 
\left( S_{\mathcal{A}_n^p \mathcal{A}_n^p} - \Sigma^0_{ \mathcal{A}_n^p \mathcal{A}_n^p} \right) 
u_{\mathcal{A}_n^p}. \label{eq26}
\end{eqnarray}

By Lemma \ref{lem:consistency:1} and Lemma \ref{lem:consistency:3}, it follows that, for any $\gamma > 0$, 
there exist constants $C_{2, \gamma}$ and $C_{3, \gamma}$ such that 
\begin{align*}
& \| d_{\mathcal{A}_n^p} (\eta) \|_2\\
& \quad \geq 2 \tilde{\alpha}_n \left\| \Sigma^0_{ \mathcal{A}_n^p \mathcal{A}_n^p} u_{\mathcal{A}_n^p} \right\|_2 - 
C_{2, \gamma} \sqrt{\frac{d_n \log n}{n}} - C_{3, \gamma} 
\|u\|_2 \frac{\tilde{\alpha}_n d_n \sqrt{\log n}}{\sqrt{n}}\\
& \quad \geq \frac{\tilde{\alpha}_n}{\lambda_{\max}(\Omegatrue)} \|u\|_2\\
& \quad = \sqrt{d_n} \lambda_{1,n} 
\frac{C_1}{\lambda_{\max}(\Omegatrue)}, 
\end{align*}
with probability at least $1 - O(n^{-\gamma})$ for large enough $n$. The last inequality follows by 
condition (iii) and since $d_n (\log n / n)^{1/2} \to 0$.

Choosing $C_1 = \lambda_{\max}(\Omegatrue) + 1$ leads to the required result. 
\end{proof}

The next lemma establishes estimation and model selection (sign) consistency for the restricted 
minimization problem (\ref{eq20}). 
\begin{lemma} \label{lem:consistency:7}
There exists $C_2 > 0$ such that, for any $\gamma > 0$, the following holds with probability at least $1 - 
O(n^{-\gamma})$ for large enough $n$:
\begin{enumerate}[a.]
\item there exists a solution to the restricted minimization problem (\ref{eq20})
\item (estimation consistency) any global minimum of the restricted minimization problem (\ref{eq20}) lies within the ball $\{\eta: \|\eta - \eta^0\|_2 < 
C_2 \sqrt{d_n} \lambda_{1,n}\}$
\item (sign consistency) for any solution 
$\hat\eta$ of the restricted minimization problem (\ref{eq20}), $\sign \hat{\eta}_j = 
\sign \eta^0_j$, for every $1 \leq j \leq r$.
\end{enumerate}
\end{lemma}

\begin{proof}
The existence  of a solution follows from Lemma \ref{lem:consistency:6}.

By the optimality conditions 
for the restricted minimization problem (\ref{eq20}) (along the lines of Lemma \ref{lem:consistency:2}), it 
follows that, for any solution $\hat\eta$ of (\ref{eq20}), $|d_j (\hat\eta)| \leq 
\lambda_{1,n}$, for every $j \in \mathcal{A}_n^p$. It follows that $\left\| d_{\mathcal{A}_n^p} 
(\hat\eta) \right\|_2 \leq \sqrt{d_n} \lambda_{1,n}$. Estimation consistency now 
follows from Lemma \ref{lem:consistency:7}.

Note that, by condition (vi) and the fact that $d_n \leq q_n$,
$$
\eta^0_j \geq \frac{s_n}{\lambda_{\max}(\Omegatrue)} > 2 C_2 \sqrt{d_n} \lambda_{1,n},
$$
for every $j \in \mathcal{A}_n^p$ and for sufficiently large $n$. Sign consistency now follows by combining this fact with 
$\|\eta - \eta^0\|_2 < C_2 \sqrt{d_n} \lambda_{1,n}$. 
\end{proof}

The next lemma will be instrumental in showing that the solution set of the restricted minimization problem 
(\ref{eq20}) is the same as the solution set of the unrestricted minimization problem for $J_p$ 
with high probability. 

\begin{lemma} \label{lem:consistency:8}
For any $\gamma > 0$, any solution $\hat\eta$ of (\ref{eq20}) satisfies 
$$
\max_{j \notin \mathcal{A}_n^p} \left| d_j (\hat\eta) \right| < \lambda_{1,n}, 
$$
with probability at least $1 - O(n^{-\gamma})$ for large enough $n$. 
\end{lemma}

\begin{proof}
Let $\gamma > 0$ be given, and let $\hat\eta$ be a solution of (\ref{eq20}). If 
$C_n = \{ \sign \hat\eta = \sign \eta^0 \}$, then $\P(C_n) 
\geq 1 - O(n^{-\gamma-\kappa})$ for large enough $n$ (by Lemma \ref{lem:consistency:7}). Now, on 
$C_n$, it follows by a first-order expansion of $d_{\mathcal{A}_n^p}$ around 
$\eta^0$ and the optimality conditions for (\ref{eq20}), that 
\begin{eqnarray}
- \lambda_{1,n} \sign \eta^0_{\mathcal{A}_n^p} 
&=& d_{\mathcal{A}_n^p} (\hat\eta) \nonumber\\
&=& d_{\mathcal{A}_n^p} (\eta^0) + 2 S_{\mathcal{A}_n^p \mathcal{A}_n^p} \hat{u}_n 
\nonumber\\
&=& H_n \hat{u}_n + d_{\mathcal{A}_n^p} (\eta^0) + 2 \left( S_{\mathcal{A}_n^p 
\mathcal{A}_n^p} - \Sigma^0_{ \mathcal{A}_n^p \mathcal{A}_n^p} \right) \hat{u}_n, \label{eq27} 
\end{eqnarray}
where $\hat{u}_n = \hat\eta - \eta^0$, and $H_n = 2 
\Sigma^0_{ \mathcal{A}_n^p \mathcal{A}_n^p}$.

Hence,
\begin{equation} \label{eq28}
\hat{u}_n = - \lambda_{1,n} H_n^{-1} \sign \eta^0_{\mathcal{A}_n^p} - H_n^{-1} 
d_{\mathcal{A}_n^p} (\eta^0) - 2 H_n^{-1} \left( S_{\mathcal{A}_n^p \mathcal{A}_n^p} - 
\Sigma^0_{ \mathcal{A}_n^p \mathcal{A}_n^p} \right) \hat{u}_n. 
\end{equation}

Now, let us fix $j \notin \mathcal{A}_n^p$. By a first-order Taylor series expansion of $d_j$, it follows 
that 
$$
d_j (\hat\eta) = d_j (\eta^0) + 2 S_{i,\mathcal{A}_n^p}^T \hat{u}_n. 
$$

Using (\ref{eq28}), we get that 
\begin{eqnarray} \label{eq29}
d_j (\hat\eta) 
&=& d_j (\eta^0) + 2 (S_{j,\mathcal{A}_n^p} - \Sigma^0_{j,\mathcal{A}_n^p})^T 
\hat{u}_n + 2 ( \Sigma^0_{j,\mathcal{A}_n^p} )^T \hat{u}_n \nonumber\\
&=&  -2 \lambda_{1,n} ( \Sigma^0_{j,\mathcal{A}_n^p} )^T H_n^{-1} \sign \eta^0_{\mathcal{A}_n^p} + d_j (\eta^0) - 2 
( \Sigma^0_{j,\mathcal{A}_n^p} )^T H_n^{-1} d_{\mathcal{A}_n^p} (\eta^0) + 
\nonumber\\
& & - 4 ( \Sigma^0_{j,\mathcal{A}_n^p} )^T H_n^{-1} \left( S_{\mathcal{A}_n^p \mathcal{A}_n^p} - 
\Sigma^0_{ \mathcal{A}_n^p \mathcal{A}_n^p} \right) \hat{u}_n + 2 (S_{i,\mathcal{A}_n^p} - 
\Sigma^0_{i,\mathcal{A}_n^p})^T \hat{u}_n. 
\end{eqnarray}

We now individually analyze all the terms in (\ref{eq29}).

It follows, by \eqref{eq:incoh2}, that the first term satisfies 
\begin{equation} \label{eq30}
\left| -2 \lambda_{1,n} ( \Sigma^0_{j,\mathcal{A}_n^p} )^T H_n^{-1} \sign \eta^0_{\mathcal{A}_n^p} \right| \leq \delta \lambda_{1,n} < \lambda_{1,n}. 
\end{equation}

It follows, by Lemma \ref{lem:consistency:3} and since $(1 / \lambda_{1,n}) ( (d_n / n) \log n )^{1/2} \to 0$ and $d_n^{1/2} \lambda_{1,n} \to 0$, that the second term $d_j 
(\eta^0)$ is $o(\lambda_{1,n})$ with probability at least $1 - O(n^{-\gamma-\kappa})$ for 
large enough $n$.

Also, by condition (ii) and the definition of $H_n$, we get that 
\begin{equation} \label{eq31}
\left\| 2 ( \Sigma^0_{j,\mathcal{A}_n^p} )^T H_n^{-1} \right\|_2 \leq \left\| \Sigma^0_{j,\mathcal{A}_n^p} 
\right\|_2 \| 2 H_n^{-1} \|_2 \leq \frac{1}{\lambda_{\min}(\Omegatrue)} \left\| \left( \Sigma_{ \mathcal{A}_n^p \mathcal{A}_n^p}^0 \right)^{-1} 
\right\|_2 \leq \frac{\lambda_{\max}(\Omegatrue)}{\lambda_{\min}(\Omegatrue)},
\end{equation}
where $\| \cdot \|_2$ here denotes the $\ell_2$ operator norm (maximum singular value).  It follows, by Lemma \ref{lem:consistency:3} and since $(1 / \lambda_{1,n}) ( (d_n / n) \log n )^{1/2} \to 0$ and $d_n^{1/2} \lambda_{1,n} \to 0$, that the third term in (\ref{eq29}) 
satisfies 
\begin{equation} \label{eq32}
\left| 2 ( \Sigma^0_{j,\mathcal{A}_n^p} )^T H_n^{-1} d_{\mathcal{A}_n^p} (\eta^0) \right| \leq \frac{\lambda_{\max}(\Omegatrue)}{\lambda_{\min}(\Omegatrue)} \sqrt{d_n} \max_{j \in \mathcal{A}_n^p} 
|d_j (\eta^0)| = o(\lambda_{1,n}). 
\end{equation}

Let $b = 2 H_n^{-1} \Sigma_{j,\mathcal{A}_n^p}$. Note that, by (\ref{eq31}), the norm of 
$b$ is uniformly bounded in $n$ and $r$. Also note that the $j$th element of the vector $\left( 
S_{\mathcal{A}_n^p \mathcal{A}_n^p} - \Sigma^0_{ \mathcal{A}_n^p \mathcal{A}_n^p} \right) b$ 
is the difference between the sample and the population covariance of $X_j$ and $\sum_{k \in 
\mathcal{A}_n^p} b_k X_k$. Using the same line of arguments as in the proof of 
Lemma \ref{lem:consistency:3}, it follows that there exists a constant $C_{4, \gamma} > 0$ such 
that 
\begin{equation} \label{eq33}
\max_{j \in \mathcal{A}_n^p} \left| \left( \left( S_{\mathcal{A}_n^p \mathcal{A}_n^p} - \Sigma^0_{ 
\mathcal{A}_n^p \mathcal{A}_n^p} \right) b \right)_j \right| \leq C_{4, \gamma} 
\sqrt{\frac{\log n}{n}}, 
\end{equation}
with probability at least $1 - O(n^{-\gamma-\kappa})$ for large enough $n$. By (\ref{eq31}), (\ref{eq33}), 
claim (b) in Lemma \ref{lem:consistency:7}, and since $(1 / \lambda_{1,n}) ( (d_n / n) \log n )^{1/2} \to 0$ and $d_n^{1/2} \lambda_{1,n} \to 0$, we have that the 
fourth term in (\ref{eq29}) satisfies 
\begin{eqnarray} \label{eq34}
\left| 4 ( \Sigma^0_{j,\mathcal{A}_n^p} )^T H_n^{-1} \left( S_{\mathcal{A}_n^p \mathcal{A}_n^p} - 
\Sigma^0_{ \mathcal{A}_n^p \mathcal{A}_n^p} \right) \hat{u}_n \right| 
&\leq& 2 \left\| \left( S_{\mathcal{A}_n^p \mathcal{A}_n^p} - \Sigma^0_{ \mathcal{A}_n^p \mathcal{A}_n^p} 
\right) b \right\|_2 \|\hat{u}_n\|_2 \nonumber\\ 
&=& O \left( \sqrt{\frac{d_n \log n}{n}} \sqrt{d_n} \lambda_{1,n} \right) \\
&=& o(\lambda_{1,n}),
\end{eqnarray}
with probability at least $1 - O(n^{-\gamma-\kappa})$ for large enough $n$.

By Lemma \ref{lem:consistency:1}, 
claim (b) in Lemma \ref{lem:consistency:7}, and condition (ii), the fifth term in (\ref{eq29}) 
satisfies 
\begin{equation} \label{eq36}
\left| 2 (S_{i,\mathcal{A}_n^p} - \Sigma^0_{i,\mathcal{A}_n^p})^T \hat{u}_n \right| \leq 2 \left\| 
S_{i,\mathcal{A}_n^p} - \Sigma^0_{i,\mathcal{A}_n^p} \right\|_2 \|\hat{u}_n\|_2 = O \left( 
\sqrt{\frac{d_n \log n}{n}} \sqrt{d_n} \lambda_{1,n} \right) = o(\lambda_{1,n}). 
\end{equation}

It follows, by (\ref{eq29}), (\ref{eq30}), (\ref{eq32}), and (\ref{eq34})-(\ref{eq36}), that, for any $j \notin 
\mathcal{A}_n^p$, 
$$
\left| d_j (\hat\eta) \right| < \lambda_{1,n}, 
$$
with probability at least $1 - O(n^{-\gamma-\kappa})$ for large enough $n$. The result now 
follows by the union bound, and from the fact that $p = O(n^\kappa)$. 
\end{proof}

Let $\gamma > 0$ be chosen arbitrarily. Let $C_{p,n}$ denote the event on which 
Lemma \ref{lem:consistency:7} and Lemma \ref{lem:consistency:8} hold. It follows that 
$\P(C_{p,n}) \geq 1 - O(n^{-\gamma-\kappa})$, for large enough $n$. Now, on $C_{p,n}$, any solution 
of the restricted problem (\ref{eq20}) is also a global minimizer of $J_p$ (by 
Lemma \ref{lem:consistency:2}). Hence, there is at least one global minimizer of $J_p$ 
for which the components corresponding to $(\mathcal{A}_n^p)^c$ are zero. It again follows, by 
Lemma \ref{lem:consistency:2}, that these components are zero for all global minimizers 
of $J_p$. Hence, the solution set of the restricted minimization problem (\ref{eq20}) is the same as the solution set for the unrestricted problem (\ie, the set of 
global minimizers of $J_p$). Hence, on $C_{p,n}$, the assertions of 
Lemma \ref{lem:consistency:7} hold for the solutions of the unrestricted minimization problem 
for $J_p$. 

Now, let $\mathcal{B}_n^p = \mathcal{A}_n^p \cup \{p\}$. Using the sparsity in $\Omegatrue$ it can be shown that 
$\Omegatrue_{pp}$ is also the diagonal entry corresponding to the index $p$ in $\left( 
\Sigma^0_{\mathcal{B}_n^p \mathcal{B}_n^p} \right)^{-1}$. Let $\hat{\mathcal{A}}_n^p$ be the set 
of indices corresponding to the nonzero entries of any minimizer $\hat\eta$ of $J_p$, let $\hat{\Omega}_{pp}$ be the diagonal entry corresponding to the index $p$ for $\left( 
S_{\hat{\mathcal{B}}_n^p \hat{\mathcal{B}}_n^p} \right)^{-1}$, and let $\hat{\mathcal{B}}_n^p = 
\hat{\mathcal{A}}_n^p \cup \{p\}$. It follows that $\hat{\mathcal{B}}_n^p = 
{\mathcal{B}}_n^p$ on $C_{p,n}$, and that
\begin{eqnarray}
|\hat{\Omega}_{pp} - \Omegatrue_{pp}| 
&\leq& \left\| \left( S_{{\mathcal{B}}_n^p {\mathcal{B}}_n^p} 
\right)^{-1} - \left( \Sigma^0_{{\mathcal{B}}_n^p {\mathcal{B}}_n^p} \right)^{-1} \right\|_2 \nonumber\\
&\leq& \left\| \left( S_{{\mathcal{B}}_n^p {\mathcal{B}}_n^p} \right)^{-1} \right\|_2 \left\| 
S_{{\mathcal{B}}_n^p {\mathcal{B}}_n^p} - \Sigma^0_{{\mathcal{B}}_n^p 
{\mathcal{B}}_n^p} \right\|_2 \left\| \left( \Sigma^0_{{\mathcal{B}}_n^p 
{\mathcal{B}}_n^p} \right)^{-1} \right\|_2 \nonumber\\
&\leq& \lambda_{\max}(\Omegatrue) \left\| \left( S_{{\mathcal{B}}_n^p {\mathcal{B}}_n^p} \right)^{-1} \right\|_2 \left\| 
S_{{\mathcal{B}}_n^p {\mathcal{B}}_n^p} - \Sigma^0_{{\mathcal{B}}_n^p 
{\mathcal{B}}_n^p} \right\|_2 \nonumber\\
&\leq& d_n \lambda_{\max}(\Omegatrue) \left\| \left( S_{{\mathcal{B}}_n^p {\mathcal{B}}_n^p} \right)^{-1} \right\|_2 
\max_{1 \leq i,j \leq p} |S_{ij} - \Sigma^0_{ij}|. \label{eq37} 
\end{eqnarray}

Note that, by Lemma \ref{lem:consistency:1}, there exists a constant $C_{\gamma + \kappa}$ such that 
\begin{eqnarray*}
\| S - \Sigma^0_n \|_{\max} = \max_{1 \leq i,j \leq p} |S_{ij} - \Sigma^0_{ij}|  \leq 
C_{\gamma+\kappa} \sqrt{\frac{\log n}{n}}, 
\end{eqnarray*}
with probability at least $1 - O(n^{-\gamma-\kappa})$ for large enough $n$. Let $D_n$ denote the 
event on which the above inequality holds. Hence, on $D_n$, we get 
\begin{eqnarray}
\left\| \left( S_{{\mathcal{B}}_n^p {\mathcal{B}}_n^p} \right)^{-1} \right\|_2 
&\leq& \left\| \left( \Sigma^0_{{\mathcal{B}}_n^p {\mathcal{B}}_n^p} \right)^{-1} \right\|_2 + \left\| 
\left( S_{{\mathcal{B}}_n^p {\mathcal{B}}_n^p} \right)^{-1} - \left( \Sigma^0_{{\mathcal{B}}_n^p 
{\mathcal{B}}_n^p} \right)^{-1} \right\|_2 \nonumber\\
&\leq& \lambda_{\max}(\Omegatrue) + d_n \lambda_{\max}(\Omegatrue) \left\| \left( S_{{\mathcal{B}}_n^p {\mathcal{B}}_n^p} \right)^{-1} \right\|_2 
\max_{1 \leq i,j \leq p} |S_{ij} - \Sigma^0_{ij}| \nonumber\\
&\leq& \lambda_{\max}(\Omegatrue) + \lambda_{\max}(\Omegatrue) C_{\gamma+\kappa} d_n \sqrt{\frac{\log n}{n}} \label{eq38}
\end{eqnarray}
for large enough $n$. It follows, by (\ref{eq37}), (\ref{eq38}), and since $d_n (\log n / n)^{1/2} \to 0$, that on $C_{p,n} \cap D_n$ 
\begin{equation} \label{eq39}
|\hat{\Omega}_{pp} - \Omegatrue_{pp}| \leq 2 \lambda_{\max}^2(\Omegatrue) C_{\gamma+\kappa} d_n \sqrt{\frac{\log n}{n}} 
\end{equation}
for large enough $n$.

For every $1 \leq i \leq p$, the above argument can be repeated verbatim by 
considering $\eta$ to be the $i$th (off-diagonal) row of $\Omegatrue$ normalized by the 
corresponding entry, and constructing the $J_i$, $\mathcal{A}_n^i$, \etc~accordingly. Then, by maximizing 
$J_i$, we can obtain $\hat{\mathcal{A}}_n^i$ such that there exists a set $C_{i,n}$ with $\P(C_{i,n}) = 1 - O(n^{-\gamma-\kappa})$ for large enough $n$, and $\hat{\mathcal{A}}_n^i = 
\mathcal{A}_n^i$ on $C_{i,n}$. Again, it can be shown in exactly the same way as above (for the case of 
the $p$th row), that if $\hat{\Omega}_{ii}$ is the diagonal entry corresponding to the index $i$ for $\left( 
S_{\hat{\mathcal{B}}_n^i \hat{\mathcal{B}}_n^i} \right)^{-1}$, then on $C_{i,n} \cap D_n$
\begin{equation} \label{eq40}
|\hat{\Omega}_{ii} - \Omegatrue_{ii}| \leq 2 \lambda_{\max}(\Omegatrue)^2 C_{\gamma+\kappa} d_n 
\sqrt{\frac{\log n}{n}}. 
\end{equation}

It follows, by (\ref{eq39}) and (\ref{eq40}), that on $\left( \cap_{i=1}^p C_{i,n} \right) \cap D_n$
\begin{equation} \label{eq41}
\max_{1 \leq i \leq p} |\hat{\Omega}_{ii} - \Omegatrue_{ii}| \leq 2 \lambda_{\max}^2(\Omegatrue) C_{\gamma+\kappa} 
d_n \sqrt{\frac{\log n}{n}}. 
\end{equation}

Since 
$$
\P \left( \left( \cap_{i=1}^p C_{i,n} \right) \cap D_n \right) \geq 1 - (p+1) O(n^{-\gamma-\kappa}) = 1 - 
O(n^{-\gamma}) 
$$
for large enough $n$, we have achieved our goal. 

Note that the estimation accuracy in Lemma \ref{lem:consistency:7} is $\sqrt{d_n} \lambda_{1,n}$. Hence, an 
estimate of $\Omega_{pp}$ based on $\hat\eta$ has estimation accuracy larger than or 
equal to $\sqrt{d_n} \lambda_{1,n}$. Since 
$$
d_n \sqrt{\frac{\log n}{n}} = \sqrt{d_n} \sqrt{\frac{d_n \log n}{n}} = o(\sqrt{d_n} \lambda_{1,n}), 
$$
$(1 / \lambda_{1,n}) ( (d_n / n) \log n )^{1/2} \to 0$, and $d_n^{1/2} \lambda_{1,n} \to 0$, it follows that a two-step procedure gives a provably better estimation accuracy 
than direct lasso based estimates of the diagonal entries of $\Omegatrue$.